\documentclass[acmsmall]{acmart}

\usepackage{booktabs} % For formal tables
\usepackage[utf8]{inputenc}
\usepackage{url}
\usepackage{amsthm}
\usepackage{amssymb}
\usepackage{algorithmicx, algorithm}
\usepackage[noend]{algpseudocode}
\usepackage{verbatim}
\usepackage{graphicx}
\usepackage{amssymb}
\usepackage{tikz}
\usepackage{listings}
\usepackage{booktabs}
\usepackage{array}
\usepackage{chngcntr}
\usepackage{url}
\usepackage{mathtools}
\usepackage{xcolor}
\usepackage{nomencl}
\usepackage{siunitx,xinttools,xintfrac}
\usepackage{lipsum}
\usepackage{subfig}
\setcopyright{rightsretained}
\begin{document}

\setcopyright{acmcopyright}
\acmJournal{POMACS}
\acmYear{2019} \acmVolume{3} \acmNumber{1} \acmArticle{1} \acmMonth{1} \acmPrice{15.00}\acmDOI{10.1145/3311072}

% \title{NEW PROJECT - QoS adjustment }
% \title[Assessing QoS Degradation]{Assessing QoS Degradation: An Algorithmic Approach for \\Network Optimization at Scale}
\title[QoS Degradation Assessment]{Network Resilience Assessment via QoS Degradation Metrics: An Algorithmic Approach}
% \titlenote{Produces the permission block, and
% \subtitle{Extended Abstract}
% \subtitlenote{The full version of the author's guide is available as
%   \texttt{acmart.pdf} document}

\author{Lan N. Nguyen}
\affiliation{%
  \institution{University of Florida}
}
\email{lan.nguyen@ufl.edu}

\author{My T. Thai}
\authornote{My T. Thai is the corresponding author.}
\affiliation{%
  \institution{University of Florida}
}
\email{mythai@cise.ufl.edu}

% \authornote{This author is the
%   one who did all the really hard work.}
% \affiliation{%

% \affiliation{%
% \affiliation{%
% \affiliation{%
%   \country{China}
% }

% \affiliation{%


% \author{Julius P.~Kumquat}

% \renewcommand{\shortauthors}{Lan N. Nguyen et al.}

\begin{abstract}
% Motivated by networked systems in which the functionality of the network depends on vectices in the network being within a bounded distance $\mathtt{T}$ of each other, we study the \textit{Quality of Service Degradation} problem: given a network and a set of pairs, find a minimum cost of manipulation of edge weights which ensures the distance between each pair exceeds a certain threshold $\mathtt{T}$. We introduce four algorithms with worst-case performance ratios for this problem. Each of them has its own strength in trade-off between effectiveness and running time, which are illustrated both in theory and comprehensive experimental evaluation. Our algorithms could scale up to large networks with millions of edges and nodes. Also, they outperformed an intuitive heuristic solution using centrality metric and are comparable to previously developed algorithms on special case of the problem. 

This paper focuses on network resilience to perturbation of edge weight. Other than connectivity, many network applications nowadays rely upon some measure of network distance between a pair of connected nodes. In these systems, a metric related to network functionality is associated to each edge. A pair of nodes only being functional if the weighted, shortest-path distance between the pair is below a given threshold \texttt{T}. Consequently, a natural question is on which degree the change of edge weights can damage the network functionality? With this motivation, we study a new problem, \textit{Quality of Service Degradation}: given a set of pairs, find a minimum budget to increase the edge weights which ensures the distance between each pair exceeds $\mathtt{T}$. We introduce four algorithms with theoretical performance guarantees for this problem. Each of them has its own strength in trade-off between effectiveness and running time, which are illustrated both in theory and comprehensive experimental evaluation. 
\end{abstract}



% \keywords{ACM proceedings, \LaTeX, text tagging}

\maketitle

\newtheorem{mydef}{Definition}
\newtheorem{Corollary}{Corollary}
\newtheorem{Observation}{Observation}
\algdef{SE}[DOWHILE]{Do}{doWhile}{\algorithmicdo}[1]{\algorithmicwhile\ #1}%
\DeclarePairedDelimiter\ceil{\lceil}{\rceil}
\DeclarePairedDelimiter\floor{\lfloor}{\rfloor}
\newenvironment{skproof}{%
  \renewcommand{\proofname}{Proof overview}\proof}{\endproof}
\section{Introduction}

Graph connectivity is considered as an important metric on measuring the functionality of a network. Typically, the connectivity-related problems usually ask for the minimum-size set of components (nodes or edges) whose removal disconnects the target set of nodes. This consideration has led to the investigation of many forms of cutting problems in a network: \textit{e.g} the minimum cut problem, the minimum multicut problem, the sparest cut problem \cite{vazirani2013approximation} and the most recent work, the Length-Bounded Multicut (\texttt{LB-MULTICUT}) problem \cite{kuhnle2018network}. In addition, various measures based on connectivity have formed the framework for assessment of network resilience to external attacks \cite{grubesic2008comparative,sen2009region,shen2013discovery,shen2012adaptive,nguyen2013detecting,dinh2014bound,dinh2015network,dinh2015assessing, pan2018vulnerability, dinh2010approximation, mishra2014cascading}. 
% \mthai{cite Thang's work. Indeed, critical node detection were originally from us! I formulated it in 2006 to assess the network vulnerabiltiy}

However, many network applications now consider other factors when determining a network functionality in addition to connectivity. For example, in Bitcoin network, to guarantee synchronization, not only the network connectivity is required but a network is also configured in order to ensure the broadcasting time of transaction messages under several seconds \cite{apostolaki2017hijacking}. As another example, consider a time-sensitive delivery on a road network, where edge weights represents the travel time between destinations. Connectivity between a source and a destination is insufficient when a guarantee on the delivery time is required.
% such guarantees are offered by popular retailers such as Amazon or Walmart [?].

Therefore, a natural question is whether a tech-savvy attacker can damage the network functionality without impacting the connectivity? Under various forms, this kind of attacks actually is common, yet stealthy. For example, in the I-SIG system, real-time vehicle trajectory data transmitted using the CV technology are used to intelligently control the duration and sequence of traffic signals \cite{CvAttack,CVpilot,CVpilot2,checkoway2011comprehensive,koscher2010experimental,mazloom2016security,chen2018exposing}. An adversary, therefore, can compromise multiple vehicles and send malicious messages with false data (e.g., speed and location) to the I-SIG system to impact the traffic control decisions. As reported by previous works, it has been shown that even one single attack vehicle can manipulate the intelligent traffic control algorithm in the I-SIG system and cause severe traffic jams \cite{CvAttack,chen2018exposing}. To understand the severity of such attack, it is necessary to study on which roads the attackers can target to and what is the minimum number of vehicles the attackers have to compromise to cause large-scale congestions, e.g. traveling from two certain locations takes several hours longer than usual. Such attack can be for political or financial purposes, e.g. blocking traffics of business competitors \cite{CvAttack}.  
% Attackers can form groups to target multiple roads to cause large-scale congestion.}

% Since one attack vehicle can only attack one intersection, to cause larger-scale damage, attackers can form groups to attack consecutive intersections along arterial roads in an area.

% As another example, in the routing attacks on Bitcoin network or any Blockchain applications \cite{apostolaki2017hijacking}, an attacker can compromise the routers or ISP to delay the block propagation while staying completely under the radar. The impact of this attack varies relying up on the victim. If the victim is a merchant, it is vulnerable to double spending attacks. If the victim is a miner, the attack wastes its computational power. If the victim is a regular node, it will have an outdated view of the Blockchain, lowering its contribution to network resilience. The longer the compromised router or ISP delays the block, the higher the probability that these attack can be detected. Therefore, in order to damage the synchronousness of the network, attackers can compromise multiple routers or ISPs along the way a packet is delivered. 

As another example, in Bitcoin network or any Blockchain-based applications, an attacker can target to damage the consensus between copies of public ledger of major miners by delaying block propagation between them. Recent works \cite{apostolaki2017hijacking} have shown that after receiving request for a block information from another node, a Bitcoin node can have up to 20 minutes to respond. An attacker, therefore, can flood the Bitcoin nodes with too many requests or ``dust'' messages to handle, thus delay their block delivery. By flooding multiple nodes, the attacker can disrupt miners to reach consensus on a certain state of Blockchain. The impact of this attack varies relying upon the victims. If the victim is a merchant, it is vulnerable to double spending attacks \cite{DoubleSpend}. If the victim is a miner, the attack wastes its computational power \cite{pinzon2016double}. If the victim is a regular node, it will have an outdated view of the Blockchain, and thus more vulnerable to the temporal attacks which exploit the lagging in Blockchain synchronization \cite{pinzon2016double,dennis2016temporal}. Therefore, it is necessary to study which nodes are critical and how the attacker should attack such nodes (e.g. how much bandwidth consumption) to impact the Bitcoin network functionality, e.g. causing major miners several hours to reach consensus.

With this motivation, we consider the \textit{Quality of Service Degradation} (\texttt{QoSD}) problem. Given a directed graph $G$ representing a network, threshold $\mathtt{T}$ and set of pairs $S$ in $G$, the objective is to identify a minimum budget to increase the edge (or node) weights to ensure the weighted, shortest-path distance between each pair in $S$ is no smaller than $\mathtt{T}$. Intuitively, the goal of this problem is to assess how robust the network is; the greater budget to increase edge weights found, the more resilient the network is to the perturbation in terms of edge weights. In addition, the budget to increase weight of a edge in the solution provides an indication of the importance of this edge to the desired functionality. 

In the context of network reliability, Kuhnle et al. \cite{kuhnle2018network} have recently studied a special case of our problem  under the name \texttt{LB-MULTICUT}. Different to our problem, the objective of this problem is to identify a minimum set of edges whose {\bf removal} ensures the distance between each pair of nodes is no smaller than $\mathtt{T}$. Directly adopting the \texttt{LB-MULTICUT} solutions to our \texttt{QoSD} problem is not feasible since most of those solutions exploited a trait that their problems can be formulated by Integer Programming and exhibit submodular behaviors. \texttt{QoSD} problem, on the other hand, is shown to be neither submodular nor supermodular, making \texttt{QoSD} more challenging to devise an efficient algorithm. Also, modern networked systems are increasingly massive in scale, often with size of millions of vertices and edges. The need for a scalable algorithm on large-scale networks poses another challenge for our problem. Motivated by these observations, the main contribution of this work are as follows.

% Another challenge is that modern networked systems are increasingly massive in scale, often with sized of millions of vertices and edges. In the experiment evaluation, we observe that even a simple heuristic algorithm based on the centrality metric of edges can not scale to a network with size of more than 10 million of edges. Motivated by these observation, the main contribution of this work are.
\begin{itemize}
\item We provide three highly scalable algorithms for our problem: Two iterative algorithms, $\mathtt{IG}$ and $\mathtt{AT}$, with approximation ratio $O(\gamma^{-1}(\ln \mathtt{T} + \mathtt{h} \ln n))$ and $O(\ln \mathtt{T} + \mathtt{h} \ln n)$ respectively, where $\gamma$ is a metric measuring the concave property of edge weight functions w.r.t a budget to increase edge weights, \texttt{h} is the maximum number of edges of a path connecting between a pair in $S$, and $n$ is the number of nodes in $G$; and $\mathtt{SA}$, a probabilistic approximation algorithm returning $O(\frac{\ln \mathtt{T} + \mathtt{h} \ln \mathtt{d}}{\gamma (1-e^{-\gamma}) (1-\epsilon)})$ approximation result with high probability, where $\mathtt{d}$ is the maximum degree of $G$. 
% \mthai{d?}
\item When the edge weight functions are linear w.r.t the cost to increase edge weight, we propose $\mathtt{LR}$, a randomized rounding algorithm based on LP relaxation of the problem. $\mathtt{LR}$ provides $O(\mathtt{h} \ln n)$ approximation guarantee.%, which is the best ratio we have obtained. 
% \mthai{Can I say: which is the best ratio one can ever be obtained? What I meant here: the best we obtained doesmnt mean anything! But the best one can ever be obtained mean this is the tight ratio.}
\item We extensively evaluate our algorithms on both synthetic networks and large-scale, real-world networks. All of our four algorithms are demonstrated to scale to networks with millions of nodes and edges in under a few hours and return nearly optimal solutions. Also, the experiments show the trade-off between our proposed algorithms in terms of runtime and quality of solution.
\end{itemize}

% In Bitcoin network, to guarantee the synchronousness among nodes, the network was initially configured to be very dense in order to guarantee any transaction can be broadcasted in several seconds. 

% \mthai{I taught yo this already. Must have the diversity of sentences. I shown you how to write the organization in the infocom.}
% \textit{Organization.} The rest of this paper is organized as follows. In Section \ref{sec:related}, we mention several literatures related to our problem. In Section \ref{sec:problem}, we formally define the problem and discuss the challenges. In Section \ref{sec:igta}, \ref{sec:sa}, \ref{sec:lr} and \ref{sec:discussion}, we present and analyze \texttt{IG}, \texttt{TA}, \texttt{SA} and \texttt{LR} respectively. In Section \ref{sec:experiment}, we evaluate our algorithms and compare to a heuristic methods and previously developed algorithms on special case of our problem. Section \ref{sec:conclusion} concludes the paper.

% \mthai{I rewrote it here}
\textit{Organization.} The rest of this paper is organized as follows. Section \ref{sec:related} reviews literatures related to our problem. In Section \ref{sec:problem}, we formally define the problem and discuss its challenges. The four solutions, \texttt{IG}, \texttt{AT}, \texttt{SA} and \texttt{LR}, are presented in Section \ref{sec:igta}, \ref{sec:sa}, \ref{sec:lr} and \ref{sec:discussion}, respectively. In Section \ref{sec:experiment}, we evaluate our algorithms, comparing to heuristic methods for the general case and to algorithms in \cite{kuhnle2018network} for the special case. Finally, Section \ref{sec:conclusion} concludes the paper.

\section{Related works} \label{sec:related}
% \mthai{I have edited a few in here but not completely yet. The structure is fine. Let that center check the language/gram}

\textbf{Relationship with Kuhnle et al.} \cite{kuhnle2018network} Kuhnle et al. has studied the Length-Bounded Multicut Problem (\texttt{LB-MULTICUT}). The objective of this problem is to identify a minimum set of edges whose {\bf removal} ensures the distance between each pair of nodes of a given set $S$ is no smaller than $\mathtt{T}$. \texttt{LB-MULTICUT} is a special case of \texttt{QoSD} where we restrict to two conditions: 1) the only way to increase an edge weight is making the weight greater than \texttt{T} and 2) the cost of doing so is uniform among edges.

%\textcolor{blue}{In our work, we generalize the \texttt{LB-MULTICUT} problem. If the only way to increase an edge weight is making the weight greater than \texttt{T} and the cost of doing so is uniform among edges, our problem becomes \texttt{LB-MULTICUT}, then any of the algorithms proposed in \cite{kuhnle2018network} can be used to solve this special case of our problem. }

Our \texttt{QoSD} problem is more general and realistic than \texttt{LB-MULTICUT}, as briefly discussed earlier. In the adversarial perspective, it is impractical to remove edges out of a network structure. Taking the I-SIG system as an example, the attacker can only damage the network functionality by compromising multiple vehicles, causing severe traffic jams on road network rather than physically damaging road lines. Furthermore, on the Bitcoin-based applications, the Bitcoin protocol only allows a maximum delay of 20 minutes for any packet delivery. For any damage of a P2P connection, the protocol creates another connection to guarantee the connectivity of Bitcoin network. Thus, the \texttt{LB-MULTICUT} cannot be applied on those two applications.

%\texttt{LB-MULTICUT} is proven to be NP-hard and the hardness result for the \texttt{LB-MULTICUT} problem is lower bounded by $\Omega(\mathtt{T)}$ under the assumption $NP \not\subseteq BPP$. 

Other than the special case, \texttt{LB-MULTICUT} and \texttt{QoSD} are fundamentally different, thus solutions to \texttt{LB-MULTICUT} are not readily applied to \texttt{QoSD}. More specifically, Kuhnle {\em et al.} proposed three approximation algorithms for \texttt{LB-MULTICUT}, which are $\mathtt{MIA}$, $\mathtt{TAG}$, $\mathtt{SAP}$ \cite{kuhnle2018network}. We are going to discuss the limits of these algorithms w.r.t solving \texttt{QoSD}.

The general idea of \texttt{MIA} is to find the multicut of sub-graphs of the input network such that each optimal multicut is a lower bound of the optimal solution of \texttt{LB-MULTICUT} instance. In this solution, the authors exploit the similarity between \texttt{LB-MULTICUT} and the multicut problem where cutting an edge in a single path is sufficient to disconnect this path. With the multicut solution, \texttt{MIA} utilizes the $O(n^{11/23})$ approximation algorithm proposed by Agarwal et al. \cite{agarwal2007improved}. Thus \texttt{MIA}'s performance guarantee is bounded by $O(Mn^{11/23})$ where $M$ is the number of considered subgraphs. Our problem does not require edge removals, so there is not clear connection with multicut. Therefore, we find it infeasible to apply \texttt{MIA}, even with modification, to solve our problem.
 
The next algorithm of \texttt{LB-MULTICUT} is \texttt{TAG}. In general, \texttt{TAG} is a dynamic algorithm, which uses a primal-dual solution to bound the worst-case performance under incremental graph changes and improves the solution in practice by periodic pruning. \texttt{TAG} utilizes the trait that cutting all edges, which are in the maximal set of disjoint paths connecting target pairs of nodes, is sufficient to disconnect those pairs. However, this solution may not be practical in our problem. Increasing weights of those edges to maximum does not guarantee the shortest paths, which connect target pairs of nodes, no smaller than the threshold $\mathtt{T}$.

The $\mathtt{SAP}$ algorithm is a greedy, sampling-based solution with an $O(\mathtt{h} \log n)$ approximation guarantee ($\mathtt{h}$ is the maximum number of edges of a single path connecting a pair in $S$), which holds with the probability of at least $1-1/m$. Our algorithm $\mathtt{SA}$ is inspired by $\mathtt{SAP}$ in that we also use a greedy approach based on path samples, generated by using probabilistic hints based upon shortest path computations to guide the sampling. However, since our objective function is non-submodular, we prove that an approximation guarantee of $\mathtt{SA}$ depends on $\gamma$, where $\gamma$ measures the concave property of edge weight functions. Moreover, we boost the process of obtaining a feasible solution by allowing a finite budget of at most $q$ to be added on each step of sampling, where $q$ can be any number. We prove that $q$ does not impact the performance guarantee of $\mathtt{SA}$.

\textbf{Optimization on Integer Lattice.} As there is a finite budget to increase the edge weight, we model our problem in a form of minimization problem on Integer Lattice: given a set of functions $\{f_i | f_i : (\mathbb{Z}^+ \cup \{0\})^n \rightarrow \mathbb{R}^+\}$ on the Integer Lattice, the objective is to minimize the cardinality of $\mathbf{x}$ that $f_i(\mathbf{x}) \geq \theta_i$ for all $i$. The optimization on the Integer Lattice has received much attention recently. However, most of those works focus on the maximization version, which asks for maximizing $f(\mathbf{x})$ under a cardinality constraint $||\mathbf{x}|| \leq k$. When $f$ is non-submodular, those works exploits either the submodularity ratio $\gamma_s$ \cite{das2011submodular}, generalized curvature $\alpha$ \cite{bian2017guarantees} or the diminishing-return ratio $\gamma_d$ \cite{kuhnle2018fast,lehmann2006combinatorial} to devise approximation solutions with performance guarantee in terms of those parameters. However, the fact that those parameters can be small and computationally hard to obtain on several real-world objectives raises a concern on those theoretical approximation ratios. For example, Kuhnle et al. \cite{kuhnle2018fast} proposed a fast maximization of Non-Submodular, Monotonic Functions on the Integer Lattice with approximation ratio $(1-e^{-\gamma_d \gamma_s} -\eta)$ for any $\eta > 0$. If $\gamma_d$ or $\gamma_s$ is $0$, this ratio will be smaller than $0$. In our work, we utilize the concave property of edge weight functions to introduce the \textit{concave ratio} $\gamma$, which we use to prove the theoretical guarantee of $\mathtt{IG}$ and $\mathtt{SA}$, and bound the sampling size of $\mathtt{SA}$. $\gamma$ can be found easily from the derivative of edge weight functions or scanning through all edge weight functions with $O(m)$ time complexity. $\gamma$ can be small in some cases, so we devise the $\mathtt{AT}$ solution from an improved $\mathtt{IG}$ algorithm, which discards the dependence on $\gamma$ value to obtain better theoretical performance guarantee but a worse runtime in trade-off. 

\textbf{Classical Multicut Problem.} The Multicut problem asks for the minimum number of edges (or nodes) whose removal ensures each pair in $S$ is topologically disconnected. For the edge version in an undirected graph, an $O(\log k)$ approximation was developed by Garg et al. \cite{garg1996approximate} by considering multicommodity flow. In directed graphs, Gupta \cite{gupta2003improved} developed an $O(\sqrt[]{n})$ approximation algorithm, which was later improved to $O(n^{11/23})$ by Agarwal et al. \cite{agarwal2007improved}. These solutions were based on the optimal solution of the linear relaxation modeling the problem instance. Our \texttt{LR} algorithm was inspired by this approach but we have to deal with the challenge that a LP-optimal value of each edge could be larger than 1. Therefore, any discretization technique of the Multicut problem cannot be directly applied to our problem. We have devised a randomized rounding technique on which we can obtain a feasible solution with high probability while ensuring an $O(\mathtt{h} \log n)$ performance ratio.

% we extended the constraints that each variable can be any real value instead of in range $[0,1]$. Then the discretization of LP solution is provided with its approximation guarantee.   
% \mthai{This last paragraph needs to strengthen... Look xo xai, qua loa qua}

\section{Problem Formulation} \label{sec:problem}

% \mthai{Strucutre is fine. Full of errors. I (as a fresh reviewer) don't even understand what $x$ is.}

In this section, we formally define the \textit{Quality of Service Degradation} ($\mathtt{QoSD}$) problem in the format of cardinality minimization on the Integer Lattice and present challenges on solving $\mathtt{QoSD}$.  

We abstract the network using a weighted directed graph $G=(V,E)$ with $|V|=n$ nodes and $|E| = m$ directed edges. Each edge $e$ is associated with a function $f_e: \mathbb{Z}^\geq \rightarrow \mathbb{Z}^+$ which indicates the weight of $e$ w.r.t a budget to increase weight of $e$. In another word, if we spend $x$ on edge $e$, the weight of edge $e$ will be $f_e(x)$. $f_e$ is monotonically increasing.

% it costs $x$ to increase the weight of edge $e$ to $f_e(x)$. $f_e$ is monotonically increasing.

% the cost of manipulation $x$.
% \mthai{what does it mean to manipulation $x$. So $x$ is an edge. Of course not!!!!}
Let $b_e$ be the maximum possible budget to increase the weight of edge $e$. Denote $\mathbf{x} = \{x_1,...x_m\}$ is a vector where $x_i$ is the budget to increase weight of the $i^\textnormal{th}$ edge and similarly $\mathtt{b} = \{b_1,...b_m\}$, we have $x_i \leq b_i$ $\forall i\in [1,m]$. $\mathtt{b}$ is called \textit{the box}. The overall budget to increase weight of all edges is denoted by $||\mathbf{x}|| = \sum_e x_e$. Let $f=\{f_1,f_2,...f_m\}$ be a set of edge weight functions. Note that, for simplicity, the notation $e$ is used to present an edge in $E$ and also the index of this edge, i.e. if we write $x_e$, we mean the budget to increase the weight of edge $e$ (to $f_e(x)$) and also the element in $\mathbf{x}$ that is corresponding to $e$. The same rule is applied with $b_e, f_e$. Also, if we write $e-1$ (or $e+1$), we indicate the edge right next to $e$ on the left (right) in $\mathbf{x}$.

A path $p = p_0,p_1,...p_l \in G$ is a sequence of vertices such that $(p_{i-1},p_i) \in E$ for $i=1,..,l$. A path can also be understood as the sequence of edges $\{(p_0,p_1), (p_1,p_2),... (p_{k-1}, p_k)\}$. In this work, a path is used interchangeably as a sequence of edges or a sequence of nodes. %If a node $u \in p$, the path $p$ contains node $u$. If a edge $e = (u,v) \in p$, the path $p$ contain edge $e$ and two nodes $u$, $v$. 
A \textit{single path} is a path containing no cycles (i.e repeated vertices). Under a budget vector $\mathbf{x}$, the length of a path $p$ is defined as $\sum_{e \in p} f_e(x_e)$. We now formally define $\mathtt{QoSD}$ as follows:

\begin{mydef} \textnormal{Quality of Service Degradation ($\mathtt{QoSD}$).} Given a directed graph $G=(V,E)$, a set $f = \{f_e: \mathbb{Z}^\geq \rightarrow \mathbb{Z}^+ \}$ of edge weight functions, a box $\mathtt{b}$ and a target set $S = \{(s_1,t_1),...(s_k,t_k)\}$, determine a minimum budget $||\mathbf{x}||$ such that under $\mathbf{x}$, the weighted, shortest-path between each pair in $S$ exceeds a threshold $\mathtt{T}$. A problem instance may be represented by the tuple $(G,f,\mathtt{b},S,\mathtt{T})$
\end{mydef}

For each edge $e \in E$, let $w_e = f_e(0)$ denote the initial weight of $e$. In this work, we assume $w_e > 0$ for all $e \in E$, which can be justified by the fact that most networks have positive costs associated with their edges, even when there is no interference from external sources (i.e., propagation delay in communication networks,  processing delay in Blockchains). 

Let $\mathcal{P}_i$ denote a set of simple paths connecting the pair $(s_i,t_i) \in S$ and $\sum_{e \in p} w_e < \mathtt{T}$ for all $p \in \mathcal{P}_i$. Let $\mathcal{F} = \cup^k_{i=1} \mathcal{P}_i$, we call a path $p \in \mathcal{F}$ a \textit{feasible path} and $\mathcal{F}$ is a set of all feasible paths in $G$. Let $\mathtt{w} = \min_e w_e$, it is trivial that the number of edges of a feasible path is upper-bounded by $\ceil{\frac{\mathtt{T}}{\mathtt{w}}}$. Denote $\mathtt{h} = \ceil{\frac{\mathtt{T}}{\mathtt{w}}}$.
% , which is a maximum number of edges of a feasible path. 

Under $\mathbf{x}$, given a pair of nodes $(s,t)$, if there exists no single path $p$ from $s$ to $t$ which satisfies $\sum_{e \in p} f_e(x_e) < \mathtt{T}$, we call $s$ is separated from $t$ or the pair $(s,t)$ is \textit{separated} by $\mathbf{x}$. Also, given a feasible path $p \in \mathcal{F}$, if $\sum_{e \in p} f_e(x_e) \geq \mathtt{T}$, we call $p$ is \textit{blocked} by $\mathbf{x}$ or $\mathbf{x}$ blocks $p$.

The $\mathtt{QoSD}$ problem can be formulated as the follows:

\begin{align}
\min & \quad ||\mathbf{x}|| \\
\text{s.t. } & \sum_{e \in p} f_e(x_e) \geq \mathtt{T} && \forall  p \in \mathcal{F} \label{con:path} \\
& x_e \leq b_e && \forall e \in E \\
& x_e \in \mathbb{Z}^+ \cup \{0\} && \forall e \in E
\end{align}
Note that even $x_e \in \mathbb{Z}^+ \cup \{0\}$, this is not an Integer Program because $f_e(x)$ may not be a linear function. 

We can see this formulation as the cardinality minimization on the Integer lattice to satisfy multiple constraints. Before going further, we will look at several notations, mathematical operators on Integer lattice, which will be used along the theoretical proofs of our algorithms. Given $\mathbf{x} = \{x_1,...x_m\}, \mathbf{y} = \{y_1,...y_m\} \in \mathbb{Z}^m$, we have:
\begin{align*}
\mathbf{x} + \mathbf{y} & = \{x_1+y_1,...x_m+y_m\} \\
\mathbf{x} - \mathbf{y} & = \{x_1 - y_1,... x_m - y_m\} \\
\mathbf{x} \wedge \mathbf{y} & = \{ \min(x_1,y_1),...\min(x_m,y_m)\} \\
\mathbf{x} \vee \mathbf{y} &= \{\max(x_1,y_1),... \max(x_m,y_m)\} \\
\mathbf{x} / \mathbf{y} &= \{\max(x_1 - y_1,0),... \max(x_n - y_n, 0)\}\\
c \mathbf{x} & = \{cx_1,...cx_m\} \quad \forall c \in \mathbb{Z}
\end{align*}

Moreover, we say $\mathbf{x} \leq \mathbf{y}$ if $x_i \leq y_i$ for all $i \in [1,m]$, the similar rule is applied to $<,\geq,>$. 

Let $\mathbf{s}_i$ be a unit vector with the same dimension with $\mathbf{x}$, $\mathbf{s}_i$ has value $1$ in the $i^{th}$ element and $0$ elsewhere. Therefore, we could also write $\mathbf{x} = \sum_{i=1}^m x_i \mathbf{s}_i$. Table \ref{table:notation} summarizes all the notations we have so far.

% \begin{table}[t]
% \centering
% \caption{Notation} \label{table:notation}
% % \vspace*{-10px}
% \begin{tabular}{c|l}
% \toprule
%      Notation & Definition   \\ 
%  \midrule
%  $G=(V,E)$ & Input directed graph \\
%  $V,E$	& Vertex and edge sets of $G$, respectively \\
%  $n,m$	& Number of vertices, edges in $G$, respectively \\
%  \texttt{d}	& The maximum degree of $G$ \\
%  $S$	& The set of target pairs of nodes \\
%  $k$	& The number of pairs in target set $S$ \\
%  $\mathtt{T}$	& The threshold on the path length \\
%  $f_e(x)$		& The weight function of edge $e$ w.r.t a budget $x$ \\
%  $f$			& The set of all weight functions of edges in $G$ \\
% %  $w_e = f_e(0)$	& The initial weight of edge $e$ \\
%  $\mathcal{F}$	& The set of all feasible paths \\
%  $\mathtt{h}$	& The maximum number of edges of a path in $\mathcal{F}$ \\
%  $q$	& The maximum added cost in each iteration of \texttt{SA} \\
%  $\mathbf{x} = \{x_1,..x_m\}$ & The budget vector, $x_i$ is the budget on edge $i$\\
%  $\mathbf{s}_i$ & Unit vector, 1 in the $i^{\textnormal{th}}$ element and $0$ elsewhere \\
%  $\gamma$		& The concave ratio of the function set $f$ \\
%  $\alpha$		& Bias parameter in the sampling of $\mathtt{SA}$\\
%  $\mathtt{x}^*$	& Optimal solution to the problem instance\\
%  $\mathtt{OPT} = ||\mathbf{x}^*||$	& Size of optimal solution\\
% \bottomrule
% \end{tabular}
% \end{table}

\textit{Discussion.} Given an instance of $\mathtt{QoSD}$ $(G,f,\mathtt{b},S,\mathtt{T})$, the optimal solution can be obtained by formulating the problem as the following Integer Programming (IP):
\vspace*{-5px}
\begin{align}
\min & \quad \sum_e \sum^{b_e}_{i=0} i \cdot y_{e,i} \\
\text{s.t. } & \sum_{i=0}^{b_e} y_{e,i} = 1 && \forall e \in E  \label{con:cost}\\
&\sum_{e \in p} \sum_{i=0}^{b_e} f_e(i) \cdot y_{e,i} \geq \mathtt{T} && \forall  p \in \mathcal{F} \label{con:costpath} \\
& y_{e,i} \in \{0,1\} && \forall e \in E, i \in [0,b_e]
\end{align}
where $y_{e,i}$ is an indicator variable which is $1$ if $x_e = i$ and $0$ otherwise. The first constraint (Eq. \ref{con:cost}) is to guarantee the budget to increase weight of edge $e$ is a value in range $[0,b_e]$ and the second constraint (Eq. \ref{con:costpath}) is to ensure the length of each feasible path is at least $\mathtt{T}$. However, solving this IP is extremely expensive. Not only because solving IP is NP-hard (the performance is strongly dependent on which solver is used) but also listing all the paths for the second constraint is very expensive in practice since it requires $O(m^\mathtt{h})$ in the worst case. Our algorithms are designed to be efficient even when $G$ is large and hence do not require a listing of $\mathcal{F}$ or an optimal solution of the linear relaxation of this IP formulation.

\textit{Hardness and Inapproximability.} Since \texttt{LB-MULTICUT} is a special case of $\mathtt{QoSD}$, $\mathtt{QoSD}$ is NP-hard. Furthermore, any inapproximability result of \texttt{LB-MULTICUT} or the Multicut problem is also the inapproximability of \texttt{QoSD}. We summarize those results as follows:
\begin{itemize}
\item Kuhnle et al. \cite{kuhnle2018network} Let $\mathtt{T} \geq 16$. Unless $NP \subseteq BPP$, there is no polynomial-time algorithm to approximate $\mathtt{QoSD}$ within a factor of $\floor{\frac{\mathtt{T}}{6}} - 1 - \epsilon$ for any $\epsilon > 0$.
\item Lee et al. \cite{lee2016improved}: When $\mathtt{T}$ is fixed and initial edge weights are uniform, $\mathtt{QoSD}$ is inapproximable within a factor of $\Omega(\sqrt[]{T})$ assuming the Unique Games Conjecture.
\item Chawla et al. \cite{chawla2006hardness}: There exist no $O(\log \log n)$-approximation algorithm for $\mathtt{QoSD}$ unless $P=NP$. 
\end{itemize}

\textit{Node version of the problem.} The node version of the $\mathtt{QoSD}$ problem asks for the minimum budget to increase node weights rather than edge weights in the problem definition above. All our four algorithms can be easily adapted for the node version and keep the same theoretical performance guarantees.

\begin{table}[t]
\centering
\caption{Notation} \label{table:notation}
% \vspace*{-10px}
\begin{tabular}{c|l}
\toprule
     Notation & Definition   \\ 
 \midrule
 $G=(V,E)$ & Input directed graph \\
 $V,E$	& Vertex and edge sets of $G$, respectively \\
 $n,m$	& Number of vertices, edges in $G$, respectively \\
 \texttt{d}	& The maximum degree of $G$ \\
 $S$	& The set of target pairs of nodes \\
 $k$	& The number of pairs in target set $S$ \\
 $\mathtt{T}$	& The threshold on the path length \\
 $f_e(x)$		& The weight function of edge $e$ w.r.t a budget $x$ \\
 $f$			& The set of all weight functions of edges in $G$ \\
%  $w_e = f_e(0)$	& The initial weight of edge $e$ \\
 $\mathcal{F}$	& The set of all feasible paths \\
 $\mathtt{h}$	& The maximum number of edges of a path in $\mathcal{F}$ \\
 $q$	& The maximum added cost in each iteration of \texttt{SA} \\
 $\mathbf{x} = \{x_1,..x_m\}$ & The budget vector, $x_i$ is the budget on edge $i$\\
 $\mathbf{s}_i$ & Unit vector, 1 in the $i^{\textnormal{th}}$ element and $0$ elsewhere \\
 $\gamma$		& The concave ratio of the function set $f$ \\
 $\alpha$		& Bias parameter in the sampling of $\mathtt{SA}$\\
 $\mathtt{x}^*$	& Optimal solution to the problem instance\\
 $\mathtt{OPT} = ||\mathbf{x}^*||$	& Size of optimal solution\\
\bottomrule
\end{tabular}
\end{table}

% In actual network, even there is no interference from external sources, there always exist a small delay on each link/edge. For example, on computer network, this amount comes from propagation delay or transmission delay whose main cause comes from physical structure of network components. In this problem, we call it \textit{initial weight} of edges. Denote $w_0$ is the smallest initial weight of all the edges, it is trivial to see that the number of hops of a feasible path is upper-bounded by $\ceil{\frac{\mathtt{T}}{w_0}}$. Denote $\mathtt{h} = \ceil{\frac{\mathtt{T}}{w_0}}$.  

% \section{Application}

% \subsection{Inspection Point Placement on SCADA networks}

% \subsection{Delay attacks on P2P/Blockchain systems}

% \section{Complexity}

% \section{Solution}

\section{Iterative solution} \label{sec:igta}

% In this algorithm, the idea is using greedy algorithm to increase delay of paths that connect transactions. 

There are two challenging tasks to solve the \texttt{QoSD} problem. The first one is the number of feasible paths could be extremely large, thus we need to avoid listing all the feasible paths as discussed earlier. The second challenge is that the objective function of \texttt{QoSD} can be non-submodular, depending on the edge weight functions. We handle the challenges via two different algorithms: \textit{Iterative Greedy} (\texttt{IG}) and \textit{Adaptive Trading} (\texttt{AT}). After the discussion of \texttt{IG} and \texttt{AT}, we provide the theoretical analysis and approximation guarantee of both algorithms.

To tackle the first challenge, instead of listing all feasible paths of the network, we build a set $\mathcal{P}$ of candidate paths which is a subset of $\mathcal{F}$ but blocking all paths in $\mathcal{P}$ is sufficient to separate all pairs in $S$. $\mathcal{P}$ is built incrementally and iteratively. For each iteration, we find a budget vector $\mathbf{x} = \{x_1,..x_m\}$ to block all paths in $\mathcal{P}$. Then, we set the length of an edge $e$ to be $f_e(x_e)$. Next, we check whether $\mathbf{x}$ is sufficient to separate all pairs in $S$ by checking whether there exists the shortest path of a certain pair in $S$ whose length is smaller than $\mathtt{T}$. If yes, then blocking all paths in $\mathcal{P}$ is not sufficient to separate all pairs in $S$; we add all the shortest paths of pairs whose length has not exceeded $\mathtt{T}$ into $\mathcal{P}$ and continue to the next iteration. If no, then $\mathbf{x}$ is sufficient to separate all pairs in $S$; we terminate the algorithm and return $\mathbf{x}$. The full algorithms is represented by Alg. \ref{alg:iterative}.  

% we find a solution $\mathbf{x}$ to block a subset of feasible paths and check whether $\mathtt{x}$ is sufficient to block all given pairs. 

% a subset of paths and try to block all these paths. Then we will check whether our current solution can separate all transaction. If not, we will expand the subset of paths and find a minimal solution to block all paths on this subset again until our termination condition is met. To check whether our current solution could separate all transaction, we find the shortest path between transaction where the weight of a edge are corresponding to this edge's level of our current solution. If there exists the shortest path that smaller than $\mathtt{T}$, then our termination condition has not been met yet. So we expand the subset of paths by adding the shortest paths of transactions. The overall algorithm is presented as Alg. \ref{alg:iterative}.

\begin{algorithm}[t]
	\caption{Iterative Solution}
    \label{alg:iterative}
	\begin{flushleft}
		\textbf{Input} $G, f, \mathtt{b}, \mathtt{T},S$\\
		\textbf{Output} QoS adjustment vector $\mathbf{x}$
	\end{flushleft}
    \begin{algorithmic}[1]
    	\State $\mathcal{P} = \emptyset$
		\While{There exists path $p \in \mathcal{F}$ whose length $< \mathtt{T}$} \label{line:outer_it}
        	\State $\mathcal{P} \leftarrow \mathcal{P} \cup $ potentialPaths($G,S,\mathbf{x}$). \label{line:add_path}
            \State $\mathbf{x} = \{0\}^m$
            \State Find $\mathbf{x}$ to block all paths in $\mathcal{P}$ \label{line:block_path}
        \EndWhile
     \end{algorithmic}
     \begin{flushleft}
     	\textbf{Return $\mathbf{x}$}
     \end{flushleft}
\end{algorithm}

% \begin{Corollary}
% The algorithm \ref{alg:iterative} returns feasible solution.
% \end{Corollary}
% This corollary is trivial since our algorithm only terminates when there is no paths whose length $\leq \mathtt{T}$ given solution $\mathbf{x}$, which also means $\mathbf{x}$ is enough to block all feasible paths connecting a transaction in $S$. 

% \begin{Corollary}
% The number of \textit{while} iteration in the algorithm \ref{alg:iterative} is bounded by $n^\mathtt{h}$.
% \end{Corollary}

Since the maximum number of edges of a feasible path could reach up to $\mathtt{h} = \floor{\frac{\mathtt{T}}{\mathtt{w}}}$, the number of feasible paths of the network $G$ is upper bounded by $O(n^\mathtt{h})$. Because we guarantee there should be at least a feasible path is added into $\mathcal{P}$ in each iteration (line \ref{line:add_path} Alg. \ref{alg:iterative}), the number of iterations in Alg. \ref{alg:iterative} is at most $O(n^\mathtt{h})$. This is a large number and comparable to the case if we tried to enumerate all feasible paths. However in experiment, we found that the number of iterations is much smaller even on large and highly dense networks. 

% This corollary comes from observation that the number of hops of a feasible path in $G$ is bounded by $\mathtt{h}$, which leads to the maximum number of feasible paths is bounded by $n^\mathtt{h}$. Since in each iteration we guarantee there should be at least a feasible path is added into the set of potential path, the number of iterations will never exceed $n^\mathtt{h}$. This number, again, is enormous and it would be comparable to the case if we try to enumerate all feasible paths. However, it is very hard for the algorithm to reach this number even the input graph is a clique because the inter-relation between paths. In our experiment, the number of iterations is much smaller even on large and dense networks.

\begin{algorithm}[H]
	\caption{\textsc{potentialPaths}$(G,S,\mathbf{x})$}
	\label{alg:shortest_path}
    \begin{flushleft}
    	\textbf{Input} $G,S, \mathbf{x}$\\
	\textbf{Output} Set $\mathcal{P}$ of paths whose lengths is smaller than $\mathtt{T}$
    \end{flushleft}
    \begin{algorithmic}[1]
		\State Assign edge $e$ length is $f_e(x_e)$ $\forall$  $e \in E$
        \For{each pair $(s,t) \in S$}
        	\State $p \leftarrow $ shortest path between $s$ and $t$
            \If{length of $p$ is smaller than $\mathtt{T}$}
            	\State $\mathcal{P} = \mathcal{P} \cup p$
            \EndIf
        \EndFor
    \end{algorithmic}
    \begin{flushleft}
    	\textbf{Return } $\mathcal{P}$
    \end{flushleft}
\end{algorithm}

\begin{lemma}
The approximation guarantee of Alg. \ref{alg:iterative} equals to the approximation guarantee of the algorithm that finds $\mathbf{x}$ to block all paths in $\mathcal{P}$
\end{lemma}
\begin{proof}
Since $\mathcal{P}$ is a subset of all feasible paths in $G$, the optimal solution to block all feasible paths is also a feasible solution to block all paths in $\mathcal{P}$. Therefore, the optimal solution to block all paths in $\mathcal{P}$ is at most the size of the optimal solution of \texttt{QoSD}. Denote $\mathbf{x}^o$ and $\mathbf{x}^*$ as the optimal solutions to block paths in $\mathcal{P}$ and $\mathcal{F}$ respectively. Assume the algorithm in line \ref{line:block_path} of Alg. \ref{alg:iterative} returns $\alpha$-approximation result. We have $||\mathbf{x}|| \leq \alpha \cdot ||\mathbf{x}^o|| \leq \alpha \cdot ||\mathbf{x}^*||$. And since finally $\mathbf{x}$ is a feasible solution to our problem, then the output $\mathbf{x}$ of Alg. \ref{alg:iterative} is within $\alpha$  factor to optimal solution $\mathbf{x}^*$.   
\end{proof}

Now let us discuss the the second challenge: how to block all paths in $\mathcal{P}$, line \ref{line:block_path} of Alg. \ref{alg:iterative}. To address this, we propose two algorithms, \textit{Greedy} and \textit{Adaptive Trading}. Before delving into the details of each algorithm, we introduce the parameter $\gamma$, which is used to measure the concave property of weight functions. $\gamma$ would be utilized on performance analysis for our algorithms. 

% Now, let take a look at how we find the QoS adjustment $\mathbf{x}$ to block all paths in $\mathcal{P}$. We devise two solutions, which is \textit{greedy} and \textit{trunk-adding}. Before we dig into detail each of them, let take a look at the concave property of edges' delay function, in which we would utilize to do approximation analysis for our algorithms. 

\subsection{Concave property of weight functions}

The concave ratio of a set of functions is defined as follows:

\begin{mydef}
(\textnormal{Concave ratio}) The concave ratio of a set $F$ of non-negative functions is the largest scalar $\gamma \in [0,1]$ such that:
\begin{align}
f(x + 1) - f(x) \geq \gamma \cdot \big(f(y + 1) - f(y) \big)  
\end{align}
For all $f \in F$ and $0 \leq x \leq y$
\end{mydef}

In our problem, the set of non-negative functions contains all weight functions of edges in $G$. Therefore, for simplicity, we denote $\gamma$ as the \textit{concave ratio} of these set of weight functions. Now, we will utilize $\gamma$ to get several useful exploration for our solutions. First, given a path $p$ and a vector $\mathbf{x}$, define: 
\begin{align}
\mathtt{r}(p, \mathbf{x}) = \min(\mathtt{T}, \sum_{e \in p} f_e(x_e)) \label{equ:define_r}
\end{align}
% \begin{mydef}
% (\textnormal{Delay Convex ratio}) The delay convex ratio of a set $F$ of non-negative delay function is the smallest scalar $\alpha \in [0,1]$ such that:
% \begin{align}
% f_e(y + 1) - f_e(y) \geq (1-\alpha)( f_e(x + 1) - f_e(x))  
% \end{align}
% For all $e \in E$ and $0 \leq x \leq y$
% \end{mydef}

Let $g(\mathcal{P},\mathbf{x})$ be an arbitrary linear combination of $\mathtt{r}(p,\mathbf{x})$ for all $p \in \mathcal{P}$. $g(\mathcal{P},\mathbf{x})$ could be presented as follows:
\begin{align}
g(\mathcal{P}, \mathbf{x}) = \sum_{p \in \mathcal{P}} \beta_p \mathtt{r}(p, \mathbf{x}) \quad \quad \beta_p \in \mathbb{R}^+ \textnormal{ } \forall \textnormal{ } p \in \mathcal{P}
\end{align}
Given a vector $\mathbf{z}$, define:
\begin{align}
\Delta_{\mathbf{z}} g(\mathcal{P},\mathbf{x}) = g(\mathcal{P}, \mathbf{x} + \mathbf{z}) - g(\mathcal{P}, \mathbf{x}) \label{equ:delta_g}
\end{align}
We have the following lemma.

\begin{lemma} \label{lemma:concave}
Given two budget vectors $\mathbf{x}, \mathbf{y}$ where $\mathbf{x} \leq \mathbf{y}$ and a unit vector $\mathbf{s}$, we have:
\begin{align*}
\Delta_\mathbf{s} g(\mathcal{P}, \mathbf{x}) \geq \gamma \Delta_\mathbf{s} g(\mathcal{P}, \mathbf{y})
\end{align*}
\end{lemma}

\begin{skproof}
Without lost of generality, we assume $\mathbf{s} = \mathbf{s}_i$, a unit vector which has value $1$ at the $i^\textnormal{th}$ element and $0$ elsewhere. We prove that: given a feasible path $p$, the marginal gain of $\mathtt{r}(p,\mathbf{x})$ by $\mathbf{s}_i$ is at least $\gamma$ times the marginal gain of $\mathtt{r}(p, \mathbf{y})$ by $\mathbf{s}_i$. 

By definition, the $\mathtt{r}(p,\cdot)$ value of any budget vector cannot exceed $\mathtt{T}$. Also, $\mathtt{r}(p,\mathbf{u}) \leq \mathtt{r}(p, \mathbf{v})$ if $\mathbf{u} \leq \mathbf{v}$. Therefore, we consider three different cases: (1) $\mathtt{r}(p,\mathbf{x}) < \mathtt{r}(p, \mathbf{x} + \mathbf{s}_i) < \mathtt{T}$; (2) $\mathtt{r}(p, \mathbf{x}) < \mathtt{r}(p, \mathbf{x} + \mathbf{s}_i) = \mathtt{T}$; and (3) $\mathtt{r}(p, \mathbf{x}) = \mathtt{r}(p, \mathbf{x} + \mathbf{s}_i) = \mathtt{T}$. All three cases guarantee $\mathtt{r}(p,\mathbf{x} + \mathbf{s}_i) - \mathtt{r}(p, \mathbf{x}) \geq \gamma \big( \mathtt{r}(p,\mathbf{y} + \mathbf{s}_i) - \mathtt{r}(p, \mathbf{y}) \big)$. Since $g(\mathcal{P},\mathbf{x})$ is a linear combination of $\mathtt{r}(p, \mathbf{x})$, the lemma follows.
\end{skproof}

\begin{lemma} \label{lemma:concave_3}
Given three budget vectors $\mathbf{x},\mathbf{y}, \mathbf{z}$ where $\mathbf{x} \leq \mathbf{y}$ we have:
\begin{align*}
\Delta_{\mathbf{z}} g(\mathcal{P}, \mathbf{x}) \geq \gamma \Delta_{\mathbf{z}} g(\mathcal{P}, \mathbf{y})
\end{align*}
\end{lemma}

\begin{proof}
Let $\mathbf{z} = \sum_{i=1}^{||\mathbf{z}||} \mathbf{s}_i$ where $\mathbf{s}_i$ is a unit vector, we have:
\begin{align*}
&\Delta_\mathbf{z} g(\mathcal{P}, \mathbf{y}) = \sum_{j=1}^{||\mathbf{z}||} \Delta_{\mathbf{s}_j} g(\mathcal{P},\mathbf{y} + \mathbf{s}_1 + ... + \mathbf{s}_{j-1}) \leq \frac{1}{\gamma} \Big( \sum_{j=1}^{||\mathbf{z}||} \Delta_{\mathbf{s}_j} g(\mathcal{P},\mathbf{x} + \mathbf{s}_1 + ... + \mathbf{s}_{j-1}) \Big) \leq \frac{1}{\gamma} \Delta_{\mathbf{z}} g(\mathcal{P}, \mathbf{x})
\end{align*}
which completes the proof.
\end{proof}

Note that $\mathtt{r}(p, \mathbf{x}) \leq \mathtt{T}$. A budget vector $\mathbf{x}$ is sufficient to block all paths in $\mathcal{P}$ iff $\mathtt{r}(p, \mathbf{x}) = \mathtt{T}$ for all $p \in \mathcal{P}$. Therefore, to block all paths in $\mathcal{P}$, we find the minimum $||\mathbf{x}||$ such that:
\begin{align}
\mathtt{D}(\mathcal{P}, \mathbf{x}) = \sum_{p \in \mathcal{P}} \mathtt{r}(p, \mathbf{x}) = |\mathcal{P}| \cdot \mathtt{T}
\end{align}
In the next subsections, we devise two approximation algorithms to find such $\mathbf{x}$ and provide their performance guarantees.

\subsection{Iterative Greedy algorithm} 
The first algorithm to block all paths in $\mathcal{P}$ is the \textit{iterative greedy} algorithm (\texttt{IG}). The general idea is that: we iteratively add a unit vector $\mathbf{s}$ into $\mathbf{x}$, which maximizes the marginal gain $\Delta_\mathbf{s} \mathtt{D}(\mathcal{P}, \mathbf{x})$, until $\mathbf{x}$ is sufficient to block all paths in $\mathcal{P}$. Hence, the final overall budget ($||\mathbf{x}||$) is equal to the number of iterations of the algorithm. \texttt{IG} is fully presented by Alg. \ref{alg:greedy_blocking}. 

However, the objective function $\mathtt{D}(\mathcal{P}, \mathbf{x})$ is neither submodular nor supermodular w.r.t $\mathbf{x}$. If each edge weight function is concave, $\mathtt{D}(\mathcal{P},\cdot)$ exhibits a \textit{submodular behavior}. On the other hand, if each weight function is convex, then $\mathtt{D}(\mathcal{P}, \mathbf{x})$ can be much more than the sum of $\mathtt{D}(\mathcal{P}, \mathbf{s})$ values of unit vectors $\mathbf{s}$ constituting $\mathbf{x}$, which is a \textit{supermodular behavior}. The non-submodularity of $\mathtt{D}(\mathcal{P}, \cdot)$ means that the $\mathbf{x}$ returned by \texttt{IG} may not have an $O(\log n)$ approximation ratio. Actually the concave ratio $\gamma$ plays an important role on the performance guarantee of \texttt{IG}, which is proved theoretically by Theorem \ref{theorem:greedy_approx} and would be further illustrated in the experimental evaluation.

% The first algorithm we propose to block all paths in set $\mathcal{P}$ is the regular greedy algorithm. The question now is what is the objective function and what approximation guarantee a greedy solution can provide. To utilize the property of delay functions we analyzed in previous sub-section, we use the objective function as follows:

% To answer the first question, let order edges in $G$ to $\{e_1,e_2,e_3,....e_m\}$ and vector $\mathbf{x} = \{x_1,x_2,...x_m\}$ denote QoS level of each edges in that order. For a path $p$, define:

% \begin{align*}
% \mathtt{r}(p, \mathbf{x}) = \min(\mathtt{T}, \sum_{e \in p} f_e(x_e))
% \end{align*}

% So, our objective function would be:

% \begin{align*}
% \mathtt{D}(\mathcal{P}, \mathbf{x}) = \sum_{p \in \mathcal{P}} \mathtt{r}(p, \mathbf{x})
% \end{align*}

% Therefore, to guarantee our solution $\mathbf{x}$ can block all paths in $\mathcal{P}$, we have to make sure $\mathtt{D}(\mathcal{P}, \mathbf{x}) = |\mathcal{P}| \mathtt{T}$. So the greedy algorithm will iteratively add a level adjustment corresponding to a unit vector to the current solution until $\mathtt{D}(\mathcal{P}, \mathbf{x})$ reaches $|\mathcal{P}| \mathtt{T}$. The fully algorithm is presented by Alg. \ref{alg:greedy_blocking}. The approximation guarantee of greedy algorithm is present by Theorem \ref{theorem:greedy_approx}.

\begin{algorithm}[h]
	\caption{Greedy blocking paths (\texttt{IG})}
    \label{alg:greedy_blocking}
    \begin{flushleft}
    \textbf{Input} $G, f, \mathtt{b}, \mathtt{T}, \mathcal{P}$\\
	\textbf{Output} a cost vector $\mathbf{x}$
    \end{flushleft}
	\begin{algorithmic}[1]
    	\State $\mathbf{x} = \{0\}^m$
		\While{$\mathtt{D}(\mathcal{P},\mathbf{x}) \leq |\mathcal{P}| \mathtt{T}$} \label{line:inner_it_greedy}
        	\For{each unit vector $\mathbf{s}$}
            	\State $\Delta_\mathbf{s} \mathtt{D}(\mathcal{P}, \mathbf{x}) = \mathtt{D}(\mathcal{P}, \mathbf{x} + \mathbf{s}) - \mathtt{D}(\mathcal{P}, \mathbf{x})$
            \EndFor
            \State $\mathbf{x} = \mathbf{x} + \textnormal{argmax}_\mathbf{s} \Delta_\mathbf{s} \mathtt{D}(\mathcal{P}, \mathbf{x})$
        \EndWhile
     \end{algorithmic}
     \begin{flushleft}
     	\textbf{Return $\mathbf{x}$}
     \end{flushleft}
\end{algorithm}

% \begin{theorem} \label{theorem:greedy_approx}
% The solution of greedy algorithm is within $O(\gamma^{-1}(\mathtt{h} \ln n + \ln \mathtt{T}))$ factor of the optimal solution for blocking all paths in $\mathcal{P}$. 
% \end{theorem}

\begin{theorem} \label{theorem:greedy_approx}
\texttt{IG} returns a solution within $O(\gamma^{-1}(\mathtt{h} \ln n + \ln \mathtt{T}))$ factor of the optimal solution for blocking all paths in $\mathcal{P}$. 
\end{theorem}

\begin{skproof}
Denote $\mathbf{x}^*$ as an optimal solution to the \texttt{QoSD} instance ($||\mathbf{x}^*|| = \mathtt{OPT}$). Denote $\mathbf{x}_i$ as our obtained solution before the $i^\textnormal{th}$ iteration in Alg. \ref{alg:greedy_blocking}. The key of our proof is that: the gap between $|\mathcal{P}| \mathtt{T}$ and $\mathtt{D}(\mathcal{P},\mathbf{x})$ will be reduced after each iteration by a factor at least $1-\frac{\gamma}{\mathtt{OPT}}$. To be specific:
\begin{align*}
|\mathcal{P}|\mathtt{T} - \mathtt{D}(\mathcal{P}, \mathbf{x}_{i+1}) \leq (1-\frac{\gamma}{\mathtt{OPT}}) (|\mathcal{P}|\mathtt{T} - \mathtt{D}(\mathcal{P}, \mathbf{x}_i))
\end{align*}
This was proved by using the property of concave ratio from lemma \ref{lemma:concave} and the greedy selection. 

Furthermore, since there should exist at least a feasible path $p \in \mathcal{P}$ such that $\mathtt{r}(p, \mathbf{x}) \leq \mathtt{T} - 1$ before the final iteration of the algorithm, we prove that the number of iterations is upper bounded by $O(\frac{\ln |\mathcal{P}| \mathtt{T}}{\gamma/\mathtt{OPT}})$. The theorem follows as the number of iterations is equal to $||\mathbf{x}||$.
\end{skproof}

\subsection{Adaptive Trading algorithm} 
The concave ratio of the edge weight functions could be very small if the weight functions are convex, which makes the approximation guarantee of \texttt{IG} undesirable. Therefore, in this section, we propose a solution whose performance guarantee does not depend on the concave ratio $\gamma$. We name this algorithm \textit{Adaptive Trading} (\texttt{AT}). 

The algorithm still works in the iterative manner and terminates only when the desired $\mathbf{x}$ is found, but different from \texttt{IG} on how the solution $\mathbf{x}$ is improved in each iteration. To be specific, in each iteration, the algorithm finds an amount of additional budget to increase the weight of an edge such that maximize the ratio between the increasing amount of $\mathtt{D}(\mathcal{P}, \mathbf{x})$ and the additional budget. Therefore, in each iteration, the additional budget could be bigger than $1$. To find such amount, the simplest way is to scan through all possible amounts of additional budget of each edge. Note that the maximum budget which can be added to increase weight of edge $e$ is upper bounded by $b_e$. Therefore, the computation complexity in each iteration of $\mathtt{AT}$ is upper bounded by $O(||\mathtt{b}||)$. Denote $\mathbf{u}(e,i) \in \mathbb{R}^m$ as a vector where the element corresponding to edge $e$ has value $i$ and other elements are $0$. \texttt{AT} is fully presented in Alg. \ref{alg:trunk_adding} and its approximation guarantee is provided by Theorem \ref{theorem:trunk_approx}.

\begin{algorithm}[h]
	\caption{Adaptive Trading solution (\texttt{AT})}
    \label{alg:trunk_adding}
    \begin{flushleft}
    \textbf{Input}  $G, f, \mathtt{b}, \mathtt{T}, \mathcal{P}$\\
	\textbf{Output} QoS adjustment vector $\mathbf{x}$
    \end{flushleft}
	\begin{algorithmic}[1]
    	\State $\mathcal{P} = \emptyset$
		\While{$\mathtt{D}(\mathcal{P},\mathbf{x}) \leq |\mathcal{P}| \mathtt{T}$} \label{line:inner_it_trunk}
        	\For{each edge $e \in E$}
            	\State $z_e = \textnormal{argmax}_{z} \frac{\Delta_{\mathbf{u}(e,z)} \mathtt{D}(\mathcal{P}, \mathbf{x})}{z}$
            \EndFor
            \State $\mathbf{x} = \mathbf{x} + \textnormal{argmax}_{\mathbf{u}(e,z_e)} \frac{\Delta_{\mathbf{u}(e,z_e)} \mathtt{D}(\mathcal{P}, \mathbf{x})}{z_e} $
        \EndWhile
     \end{algorithmic}
     \begin{flushleft}
     	\textbf{Return $\mathbf{x}$}
     \end{flushleft}
\end{algorithm}

% \begin{theorem} \label{theorem:greedy_approx}
% \texttt{IG} returns a solution within $O(\gamma^{-1}(\mathtt{h} \ln n + \ln \mathtt{T}))$ factor of the optimal solution for blocking all paths in $\mathcal{P}$. 
% \end{theorem}

\begin{theorem} \label{theorem:trunk_approx}
\texttt{AT} returns a solution within $O(\mathtt{h} \ln n + \ln \mathtt{T})$ factor of the optimal solution for blocking all paths in $\mathcal{P}$. 
\end{theorem}

\begin{skproof}
Denote $\mathbf{x}_i = \{x_1,...x_m\}$ as our obtained solution before the $i^\textnormal{th}$ iteration in Alg. \ref{alg:trunk_adding}. Let $\mathbf{x}^o_i = \{x_1^o, ... x_m^o\}$ be an optimal solution which is in addition to $\mathbf{x}_i$ to block all paths in $\mathcal{P}$. Denote $\mathbf{v}(e) = \{x_1,..x_{e-1},x_e + x_e^o,...x_m + x_m^o\}$. Trivially, $\mathbf{v}(1) = \mathbf{x}_i + \mathbf{x}^o$ and $\mathbf{v}(m+1) = \mathbf{x}_i$. Let $\mathbf{u}(e_i,j_i)$ be a vector we add into solution $\mathbf{x}_i$ in the $i^\textnormal{th}$ iteration. The key of our proof is that the following inequality is always guaranteed after each iteration.
\begin{align}
\frac{\Delta_{\mathbf{u}_i(e_i,j_i)} \mathtt{D}(\mathcal{P},\mathbf{x}_i)}{j_i} \geq \frac{\mathtt{D}(\mathcal{P},\mathbf{v}(e)) - \mathtt{D}(\mathcal{P},\mathbf{v}(e+1))}{x_e^o} \label{equ:key_trunk}
\end{align}
for any $e \in E$. This is proved by utilizing the monotonicity of $\mathtt{r}(p,\mathbf{x})$ w.r.t $\mathbf{x}$ and the trait that the selection of our algorithm ensures $\Delta_{\mathbf{u}(e, q)} \mathtt{D}(\mathcal{P}, \mathbf{x}_i) \leq \frac{q}{j_i} \Delta_{\mathbf{u}(e_i,j_i)} \mathtt{D}(\mathcal{P},\mathbf{x}_i)$ for any $e \in E$ and $q \in \mathbb{Z}^+$.

Furthermore, the Eq. \ref{equ:key_trunk} helps us to prove that: the gap between $|\mathcal{P}|\mathtt{T}$ and $\mathtt{D}(\mathcal{P}, \mathbf{x})$ will be reduced after each iteration by a factor at least $1-\frac{j_i}{\mathtt{OPT}}$. To be specific:
\begin{align*}
|\mathcal{P}|\mathtt{T} - \mathtt{D}(\mathcal{P}, \mathbf{x}_{i+1}) \leq (1 - \frac{j_i}{\mathtt{OPT}}) \big(|\mathcal{P}| \mathtt{T} - \mathtt{D}(\mathcal{P}, \mathbf{x}_i) \big)
\end{align*}
since there should exist at least a feasible path $p \in \mathcal{P}$ such that $\mathtt{r}(p, \mathbf{x}) \leq \mathtt{T} - 1$ before the final iteration of the algorithm, utilizing Cauchy theorem \cite{cauchyInequality}, we bound the budget $||\mathbf{x}||$ by $\mathtt{OPT} \cdot O(\ln |\mathcal{P}| \mathtt{T})$. Since 
$|\mathcal{P}| \leq n^\mathtt{h}$, the theorem follows.
\end{skproof}

\section{Sampling Approach} \label{sec:sa}

In this section, we introduce a sampling solution \texttt{SA} to \texttt{QoSD} which has $O(\frac{\ln \mathtt{T} + \mathtt{h} \ln d}{\gamma (1-e^-{\gamma})(1-\epsilon)})$ approximation guarantee with probability at least $1-\delta$ where $\epsilon, \delta > 0$ are arbitrarily small numbers. \texttt{SA} runs in polynomial time when the parameter $\mathtt{T}$ is fixed.

We define a \textit{blocking metric} of a budget vector $\mathbf{x}$ as follows
\begin{align*}
B(\mathbf{x}) = \sum_{p \in \mathcal{F}} \mathtt{r}(p, \mathbf{x})
\end{align*}
It is trivial that $\mathbf{x}$ blocks all pairs in $\mathcal{F}$ iff $B(\mathbf{x}) = |\mathcal{F}| \cdot \mathtt{T}$. 

In essence, \texttt{SA} attempts to minimize $||\mathbf{x}||$ while ensuring $B(\mathbf{x}) = |\mathcal{F}| \cdot \mathtt{T}$. To do so, \texttt{SA} works in the greedy manner as follows: in each iteration, \texttt{SA} finds a budget vector $\mathbf{v} = \{v_1,...v_m\}$, $||\mathbf{v}|| \leq q$, to add into $\mathbf{x}$ which maximizes $B(\mathbf{x} + \mathbf{v})$. Rather than an expensive listing of $\mathcal{F}$, an estimator is employed by path sampling procedure to find the vector $\mathbf{v}$. This process is repeated until the budget vector $\mathbf{x}$ is sufficient to block all paths in $\mathcal{F}$. \texttt{SA} is fully presented in Alg. \ref{alg:sampling_solution}.

\begin{algorithm}[h]
	\caption{Sampling Algorithm (\texttt{SA})}
    \label{alg:sampling_solution}
	\begin{flushleft}
	\textbf{Input} $G, S, \mathtt{T}, f, \mathtt{b}$ and $q, \epsilon, \delta$\\
	\textbf{Output} cost vector $\mathbf{x}$
	\end{flushleft}
    \begin{algorithmic}[1]
        \State Initiate $\mathbf{x} = \{0\}^m$
        \While{There exists a path $p \in \mathcal{F}$ whose length $< \mathtt{T}$}
%         		\State Identify $\epsilon_1, \epsilon_2, \delta_1, \delta_2$ such that $\epsilon_1 + \epsilon_2 = \epsilon$ and $\delta_1 + \delta_2 = \frac{\delta}{||\mathtt{b}||}$
            	\State Generate $\mathcal{P} = $ $\mathcal{N}(q,\epsilon, \delta/||\mathtt{b}||)$ sample paths $p_1,p_2,...$
            	\State Greedily select $\mathbf{v}$ ($||\mathbf{v}|| \leq q$) that maximizes $\hat{B}(\mathcal{P}, \mathbf{x} + \mathbf{v})$
                \State $\mathbf{x} = \mathbf{x} + \mathbf{v}$
        \EndWhile
     \end{algorithmic}
     \begin{flushleft}
     	\textbf{Return } $\mathbf{x}$
     \end{flushleft}
\end{algorithm}

Since we will not list $\mathcal{F}$, the questions now are (1) how to estimate $B(\cdot)$; and (2) how many sample paths should be generated to bound the error between the estimator of $B(\cdot)$ and its actual value. In sub-section \ref{subsec:estimate}, we define the estimator $\hat{B}(\cdot)$ employed in each iteration of Alg. \ref{alg:sampling_solution}. We provide the approximation guarantee of greedily selection on sub-section \ref{subsec:greedy_sa}. Sub-section \ref{subsec:no_sampling} provides the lower bound on the number of sampling paths to bound the error. We then put all the results together to obtain the performance guarantee of \texttt{SA}.

% Once again, ignoring all stuffs inside the while iteration, it is trivial that Alg. \ref{alg:sampling_solution} returns feasible solution to our problem. Now, we will dig into the sampling method and how we use it to incrementally build the feasible solution.

% Once again, clearly the algorithm \ref{alg:hop_increase_solution} returns feasible solution for QoS cut problem. Before going further, we define the sampling technique and what guarantee we have when we greedily select $k$ adjustment on each sampling stage.

\subsection{Estimator} \label{subsec:estimate}

Let an instance $(G,f,\mathtt{b},\mathtt{T})$ of \texttt{QoSD} be given. Denote $\mathcal{J}$ as a set of all single paths in $G$. For each $p \in \mathcal{J}$, define:

% As previous section, we still keep $\mathtt{r}(p,\mathbf{x}) = \min(\mathtt{T}, \sum_{e \in p} f_e(x_e))$ given a path $p$. In this solution, we will use different objective function. Let

% Also, denote $\mathtt{h}(p)$ is the number of hops of path $p$ and $\mathcal{F}$ is a set of all possible paths connecting transactions. Let denote $\mathcal{P}(h) = \{p \quad | \quad \mathtt{h}(p) \leq h \quad \& \quad p \in \mathcal{F}\}$. Assume $\mathcal{J}$ is set of all single paths in $G$. Given an arbitrary probability distribution $\rho$ on $\mathcal{J}$ such that $\rho(p) > 0$ for all $p \in \mathcal{P}(h,\mathbf{x})$. Define:

\begin{align*}
\mathtt{R}(p,\mathbf{x}) =
\left\{
		\begin{array}{ll}
			\mathtt{r}(p,\mathbf{x})  & \mbox{if } p \in \mathcal{F} \\
			0 & \mbox{otherwise} 
		\end{array}
	\right. \label{pf1}
\end{align*}

It is trivial that $\sum_{p \in \mathcal{J}} \mathtt{R}(p,\mathbf{x}) = \sum_{p \in \mathcal{F}} \mathtt{r}(p, \mathbf{x})$. Inspired by the estimation on the number of paths in a graph \cite{roberts2007estimating}, we define the estimator of $B(\mathbf{x})$ in the following way: Given a probability distribution $\rho$ on $\mathcal{J}$ such that $\rho(p) > 0$ for all $p \in \mathcal{F}$. Let $\mathcal{P} = \{p_1,p_2, ... p_l\}$ be a set of $l$ paths samples from $\rho$, $B(\mathbf{x})$ could be estimated by

\begin{align*}
\hat{B}(\mathcal{P},\mathbf{x}) = \frac{1}{l} \sum_{i=1}^l \frac{\mathtt{R}(p_i,\mathbf{x})}{\rho(p_i)}
\end{align*}

% We will utilize the following sampling technique which can give us the estimation of $B(\mathbf{x})$. Assume $\mathcal{J}$ is a set of all single paths in $G$. Given an arbitrary probability distribution $\rho$ on $\mathcal{J}$ such that $\rho(p) > 0$ for all $p \in \mathcal{F}$. Let $\mathcal{P} = \{p_1,p_2, ... p_l\}$ be a set of $l$ paths samples from $f$, $D(\mathbf{x})$ could be estimated by

% \begin{align*}
% \hat{D}(\mathcal{P},\mathbf{x}) = \frac{1}{l} \sum_{i=1}^l \frac{\mathtt{R}(p_i,\mathbf{x})}{\rho(p_i)}
% \end{align*}

\begin{lemma}
$\hat{B}(\mathcal{P}, \mathbf{x})$ is an unbiased estimator of $B(\mathbf{x})$
\end{lemma}

\begin{proof}
\begin{align*}
&\mathbb{E}[\hat{B}(\mathcal{P}, \mathbf{x})] = \mathbb{E}\Big[ \frac{\mathtt{R}(p,\mathbf{x})}{\rho(p)} \Big] = \sum_{p \in \mathcal{J}} \frac{\mathtt{R}(p,\mathbf{x})}{\rho(p)} \cdot \rho(p) = \sum_{p \in \mathcal{J}} \mathtt{R}(p,\mathbf{x}) = \sum_{p \in \mathcal{F}} \mathtt{r}(p, \mathbf{x})
\end{align*}
\end{proof}

% Note that the estimation of $D(\mathbf{x})$ by $\hat{D}(\mathcal{P}, \mathbf{x})$ does not depend on the probability distribution among single paths in $\mathcal{J}$ as long as the number of sample paths is sufficiently large. We will have some analyses on this later. The main point is the probability distribution among single paths does not need to be consistent between iterations. So in each iteration, our task is to find a probability distribution $\rho(\cdot)$ such that $\rho(p) > 0$ for all $p \in \mathcal{F}$. 

% \mthai{check this}

To sampling paths, we utilize the following biased, self-avoiding random walk sampling technique, which was once proposed by Kuhnle \cite{kuhnle2018network}. First, we randomly select a pair $(s,t)$ from $S$ and put $s$ into the sample path $p$. Considering in a certain moment, $p=\{s,..u\}$ ($u$ is called a tail node of $p$ at this time). The NeighborSelection procedure would select a node among the out-going neighbors of $u$ to add into $p$. The selection is as follows:  Let $\mathcal{T}$ be the shortest-path tree directed towards $t$. Let $v$ be the parent of $u$ in $\mathcal{T}$. If $N(u)/p = \{v\}$, then the next node we add into $p$ is $v$. If $v \in N(u)/p$, we select $v$ with probability $\alpha$ and the other nodes in $N(u)/p$ with probability of $\frac{1-\alpha}{|N(u)/p|-1}$. If $v \not\in N(u)/p$, we select the next node uniform randomly among $N(u)/p$. The sampling procedure ends when we meet the node $t$ or the length of $p$ exceeds $\mathtt{T}$. With the path-sampling procedure defined, given a path $p$, we could easily find $\rho(p)$. Also, $\rho(p) > 0$ for all $p \in \mathcal{F}$. The sampling technique is fully presented in Alg. \ref{alg:sampling_path}.

\begin{algorithm}[t]
	\caption{Sampling path}
    \label{alg:sampling_path}
	\begin{flushleft}
		\textbf{Input} $G, S, \mathtt{T},\mathbf{x}$ \\
		\textbf{Output} Sample path $p$
	\end{flushleft}
   	\begin{algorithmic}[1]
		\State $(s,t) \leftarrow $ randomly select a transaction. \label{line:select_pair}
        \State $p \leftarrow \{s\}$
        \Do
        	\State $u = $ tail$(p)$
            \State Let $N(u)$ be the set of outcoming neighbors of $u$
            \State $p \leftarrow p \textnormal{ } \cup $ NeighborSelection$(N(u),p, (s,t))$ \label{line:neighbor_select}
        \doWhile{$u \neq t$ and $\sum_{e \in p} f_e(x_e) < \mathtt{T}$}
     \end{algorithmic}
     \begin{flushleft}
     	\textbf{Return $p$}
     \end{flushleft}
\end{algorithm}

% \begin{lemma}
% Given a path $p \in \mathcal{F}$, the probability $p$ is selected by algorithm \ref{alg:sampling_path} are greater than $0$
% \end{lemma}

\subsection{Greedy selection on the estimator} \label{subsec:greedy_sa}
Having defined the estimator $\hat{B}(\mathcal{P},\mathbf{x})$ and the path sampling procedure, we now find the budget vector $\mathbf{v}$, $||\mathbf{v}|| \leq q$, to maximize $\hat{B}(\mathcal{P}, \mathbf{x} + \mathbf{v})$. $\mathbf{v}$ is found in the greedy manner as follows: we run in $q$ iterations and in each iteration, selecting the unit vector that maximizes the marginal gain of $\hat{B}(\mathcal{P}, \mathbf{x} + \mathbf{v})$. Since it is trivial, we will not write down the pseudo-code on how we find $\mathbf{v}$. 

The question now is what approximation guarantee $\mathbf{v}$ can provide? Note that $\hat{B}(\mathcal{P}, \mathbf{x})$ is a finite combination of functions $\mathtt{r}(p,\mathbf{x})$  with $p \in \mathcal{P}$. Hence, $\hat{B}(\mathcal{P}, \mathbf{x})$ is submodular if all weight functions are concave and supermodular if they are convex. So maximizing $\hat{B}(\mathcal{P}, \mathbf{x + v})$ using greedy algorithm may not return $1-1/e$ approximation result. Therefore, similar to \texttt{IG}, we use the concave ratio $\gamma$ to obtain the performance guarantee of the greedy selection to maximize $\hat{B}(\mathcal{P},\mathbf{x})$. 

Denote $\mathbf{v}^o = \sum_{i=1}^q \mathbf{u}_i$ as an optimal solution that maximizes $\hat{B}(\mathcal{P}, \mathbf{x} + \mathbf{v})$, where $\mathbf{u}_i$ is a unit vector ($i \in [0,m]$). Lemma \ref{lem:greedy_sampling} provides approximation guarantee of the greedy selection. 

\begin{lemma} \label{lem:greedy_sampling}
\begin{align*}
\Delta_{\mathbf{v}}\hat{B}(\mathcal{P}, \mathbf{x}) \geq (1-e^{-\gamma}) \Delta_{\mathbf{v}^o} \hat{B}(\mathcal{P}, \mathbf{x}^o)
\end{align*}
\end{lemma}

\begin{skproof}
Denote $\mathbf{v}_i$ as the budget vector $\mathbf{v}$ after greedily selecting first $i$ unit vectors. The key of the proof comes from the following inequality:
\begin{align*}
&\hat{B}(\mathcal{P},\mathbf{x} + \mathbf{v}^o) - \hat{B}(\mathcal{P},\mathbf{x} + \mathbf{v}_{i+1}) \leq (1-\frac{\gamma}{q}) \Big(\hat{B}(\mathcal{P},\mathbf{x} + \mathbf{v}^o) - \hat{B}(\mathcal{P},\mathbf{x} + \mathbf{v}_{i})\Big)
\end{align*}
This inequality is proved by using the property of $\gamma$ from lemma \ref{lemma:concave} and the trait that $\hat{B}(\mathcal{P}, \mathbf{x})$ is monotone w.r.t $\mathbf{x}$. Using this inequality, we prove that
\begin{align*}
&\Delta_{\mathbf{v}}\hat{B}(\mathcal{P},\mathbf{x}) \geq (1-(1-\frac{\gamma}{q})^q) \Delta_{\mathbf{v}^o}\hat{B}(\mathcal{P},\mathbf{x}) \geq (1-e^{-\gamma}) \Delta_{\mathbf{v}^o}\hat{B}(\mathcal{P},\mathbf{x})
\end{align*}
in which the lemma follows.
\end{skproof}

\subsection{Sample size and Performance guarantee} \label{subsec:no_sampling}
We have proved the performance guarantee of the additional budget vector $\mathbf{v}$ to maximize $\hat{B}(\mathcal{P}, \mathbf{x} + \mathbf{v})$. The question now is: what is the size of $\mathcal{P}$ to bound the error between $\Delta_\mathbf{v} \hat{B}(\mathcal{P}, \mathbf{x})$ and $\Delta_\mathbf{v} B(\mathbf{x})$? In this part, we will answer this question. Then, putting together with the performance guarantee of selecting $\mathbf{v}$ on $\mathcal{P}$, we provide the performance guarantee of \texttt{SA}.

To find the minimum number of samples, we utilize the following Chernoff Bound theory. 

% We have looked at how the greedy selection obtains approximation guarantee among sampled paths and reach a certain factor of optimum to the estimator $\hat{D}(\mathcal{P},\mathbf{x})$. The question is: what is the relationship between this local optimum and the actual objective value $D(\mathbf{x})$? The answer lies in the sample size on each sampling step. In this part, we show that: with the sufficient large number of sampling paths, we can bound the error between $\hat{D}(\mathcal{P},\mathbf{x})$ and $D(\mathbf{x})$ within an arbitrarily small value $\epsilon$ with high probability. To do so, we utilize the following Chernoff Bound theory
% In this part, we will find the number of sampling paths $\mathcal{N}(h,k,\epsilon,\delta)$ such that each $k$ QoS adjustment based on $\hat{D}(\cdot, \cdot)$ is deviated from $D(\cdot,\cdot)$ no more than $\epsilon$ with high probability. This number depends on $\mathtt{h}$, which is the hops limitation. We use the following Chernoff Bound theory

\begin{theorem}
(\textnormal{Chernoff Bound theorem} \cite{hoeffding1963probability}) Let $X_1,X_2,...X_n$ be random variables such that $a \leq X_i \leq b$ for all $i$. Let $X = \sum_{i=1}^n X_i$ and set $\mu = \mathbb{E}(X)$. Then for all $\epsilon > 0$, we have:
\begin{align}
\textnormal{Pr}[X \geq (1+\epsilon) \mu] \leq \exp(-\frac{2\epsilon^2 \mu^2}{n(b-a)^2})\\
\textnormal{Pr}[X \leq (1-\epsilon) \mu] \leq \exp(-\frac{\epsilon^2 \mu^2}{n(b-a)^2}) \label{equ:chernoff_minus} 
\end{align}
% \begin{align}
% \mathbb{P}(X \geq (1+\epsilon) \mu) \leq \exp(-\frac{2\epsilon^2 \mu^2}{n(b-a)^2}) \label{equ:chernoff_plus} \\
% \mathbb{P}(X \leq (1-\epsilon) \mu) \leq \exp(-\frac{\epsilon^2 \mu^2}{n(b-a)^2}) \label{equ:chernoff_minus} 
% \end{align}
% In short, we could write as follow:
% \begin{align}
% \mathbb{P}(|X - \mu| \geq \epsilon \mu) \leq 2 \exp(-\frac{\epsilon^2 \mu^2}{n (b-a)^2})
% \end{align}
\end{theorem}

Considering a path $p \in \mathcal{F}$, we have:
\begin{align*}
\rho(p) \geq \frac{1}{|S|} (\frac{1-\alpha}{\mathtt{d} - 1})^\mathtt{h} = \Omega(\mathtt{d}^{-\mathtt{h}} |S|^{-1})
\end{align*}
where $\mathtt{d}$ is the maximum out-going degree of a node in $G$. Therefore, for any single path $p$, $0 \leq \frac{\mathtt{R}(p,\mathbf{x})}{\rho(p)} \leq O(\mathtt{T} |S| \mathtt{d}^{\mathtt{h}})$

Denote $\mathbf{v}^*$ as an optimal solution that maximizes $\Delta_\mathbf{v} B(\mathbf{x})$. 
\begin{lemma}
Given $0<\epsilon_1,\delta_1 < 1$, with the number of sampling paths satisfies
\begin{align}
|\mathcal{P}| \geq \frac{\ln(1/\delta_1) \mathtt{T}^2 |S|^2 \mathtt{d}^{2\mathtt{h}}}{\epsilon_1^2 \Delta_{\mathbf{v}^*}^2 B(\mathbf{x})} 
\end{align}
the following condition is guaranteed:
\begin{align}
\textnormal{Pr}[\Delta_{\mathbf{v}^*} \hat{B}(\mathcal{P},\mathbf{x}) \geq (1-\epsilon_1) \Delta_{\mathbf{v}^*} B(\mathbf{x})] \geq 1 - \delta_1
\end{align}
\end{lemma}

This lemma is trivially derived from Eq. \ref{equ:chernoff_minus}.

\begin{lemma}
Given $0 < \epsilon_2, \delta_2 < 1$, with the number of sampling paths satisfies
\begin{align*}
|\mathcal{P}| \geq \frac{\ln(\binom{n+q}{q}/\delta_2) \mathtt{T}^2 |S|^2 \mathtt{d}^{2\mathtt{h}}}{2(1-e^{-\gamma})^2 \epsilon_2^2 \Delta_{\mathbf{v}^*}^2 B(\mathbf{x})}
\end{align*}
we have $\Delta_{\mathbf{v}_q} \hat{D}(x) \leq \Delta_{\mathbf{v}_q} D(\mathbf{x}) + (1-e^{-\gamma})\epsilon_2 \Delta_{\mathbf{v}^*} D(\mathbf{x})$ for all budget vectors $\mathbf{v}_k$, which satisfy $||\mathbf{v}_k|| = q$, with probability at least $1-\delta_2$
\end{lemma}

\begin{proof}
Let us consider an arbitrary budget vector $\mathbf{v}_q$, $||\mathbf{v}_q|| = q$
\begin{align*}
&\textnormal{Pr}\Big[\Delta_{\mathbf{v}_q} \hat{B}(\mathcal{P}, \mathbf{x}) \geq \Delta_{\mathbf{v}_q} B(\mathbf{x}) + (1-e^{-\gamma})\epsilon_2 \Delta_{\mathbf{v}^*} B(\mathbf{x})\Big] \\
& \quad = \textnormal{Pr}\Big[\Delta_{\mathbf{v}_q} \hat{B}(\mathcal{P},x) \geq \Delta_{\mathbf{v}_q} B(\mathbf{x}) \Big(1 + (1-e^{-\gamma})\epsilon_2 \frac{\Delta_{\mathbf{v}^*} B(\mathbf{x})}{\Delta_{\mathbf{v}_q} B(\mathbf{x})}\Big)\Big] \\
& \quad \leq \exp \Big(\frac{2(1-e^{-\gamma})^2 \epsilon_2^2 |\mathcal{P}| \Delta_{\mathbf{v}^*}^2 B(\mathbf{x})}{\mathtt{T}^2 \mathtt{d}^{2\mathtt{h}}}\Big)
\end{align*}
Using the union bound theory, to let $\Delta_{\mathbf{v}_q} \hat{B}(\mathcal{P},x) \leq \Delta_{\mathbf{v}_q} B(\mathbf{x}) + (1-e^{-\gamma})\epsilon_2 \Delta_{\mathbf{x}^*} B(\mathbf{x})$  satisfy for any budget vector $\mathbf{v}_q$, $||\mathbf{v}_q|| = k$, we have
\begin{align*}
&\textnormal{Pr}\Big[\Delta_{\mathbf{v}_q} \hat{B}(\mathcal{P}, \mathbf{x}) \geq \Delta_{\mathbf{v}_q} B(\mathbf{x}) + (1-e^{-\gamma})\epsilon_2 \Delta_{\mathbf{v}^*} B(\mathbf{x}) \Big] \leq \binom{n+q}{q} \exp \Big(\frac{2(1-e^{-\gamma})^2 \epsilon_2^2 |\mathcal{P}| \Delta_{\mathbf{v}^*}^2 B(\mathbf{x})}{\mathtt{T}^2 |S|^2 \mathtt{d}^{2\mathtt{h}}} \Big)
\end{align*}
The lemma follows by letting $\binom{n+q}{q} \exp \Big(\frac{2(1-e^{-\gamma})^2 \epsilon_2^2 |\mathcal{P}| \Delta_{\mathbf{v}^*}^2 B(\mathbf{x})}{\mathtt{T}^2 |S|^2 \mathtt{d}^{2\mathtt{h}}} \Big) \leq \delta_2$
\end{proof}

\begin{lemma} \label{cor:greedy_sa}
Given $0 \leq \epsilon_1,\epsilon_2,\delta_1,\delta_2 \leq 1$, let $\epsilon \geq \epsilon_1 + \epsilon_2$ and $\delta \geq \delta_1 + \delta_2$. If the number of sampling paths is at least
\begin{align}
\frac{\mathtt{T}^2 |S|^2 \mathtt{d}^{\mathtt{2\mathtt{h}}}}{\Delta_{\mathbf{v}^*}^2 B(\mathbf{x})} \max\Big( \frac{\ln(1/\delta_1)}{\epsilon_1^2}, \frac{\ln(\binom{n+q}{q}/\delta_2)}{2(1-e^{-\gamma})^2\epsilon_2^2} \Big) \label{equ:first_thres}
\end{align}
the greedy algorithm on $\mathcal{P}$ returns a budget vector $\mathbf{v}$ that guarantees
\begin{align*}
\textnormal{Pr}[\Delta_{\mathbf{v}} B(\mathbf{x}) \geq (1-e^{-\gamma})(1-\epsilon) \Delta_{\mathbf{v}^*}B(\mathbf{x})] \geq 1 - \delta
\end{align*}
\end{lemma}

\begin{proof}
% Given
% \begin{align*}
% \mathcal{N}(k,\epsilon,\delta) \geq \frac{\mathtt{T}^2 |S|^2 \mathtt{d}^{\mathtt{2\mathtt{h}}}}{\Delta_{\mathbf{x}^o}^2 D(\mathbf{x})} \max\Big( \frac{\ln(1/\delta_1)}{\epsilon_1^2}, \frac{\ln(\binom{n+k}{k}/\delta_2)}{(1-e^{-\gamma})^2\epsilon_2^2} \Big)
% \end{align*}
For the given number of sample paths, we have
\begin{align}
\Delta_{\mathbf{v}} B(\mathbf{x}) & \geq \Delta_{\mathbf{v}} \hat{B}(\mathcal{P},\mathbf{x}) - (1-e^{-\gamma}) \epsilon_2 \Delta_{\mathbf{v}^*} B(\mathbf{x}) \label{equ:first_equ} \\
& \geq (1-e^{-\gamma}) \Delta_{\mathbf{v}^*} \hat{B}(\mathcal{P},\mathbf{x}) - (1-e^{-\gamma}) \epsilon_2 \Delta_{\mathbf{x}^*} B(\mathbf{x}) \\
& \geq (1-e^{-\gamma})(1-\epsilon_1) \Delta_{\mathbf{v}^*} B(\mathbf{x}) - (1-e^{-\gamma}) \epsilon_2 \Delta_{\mathbf{v}^*} B(\mathbf{x}) \label{equ:second_equ} \\
& \geq (1-e^{-\gamma})(1-\epsilon) \Delta_{\mathbf{v}^*} B(\mathbf{x})
\end{align}
The inequality (\ref{equ:first_equ}) happens with probability $1-\delta_1$ while the inequality (\ref{equ:second_equ}) happens with probability $1-\delta_1$. Overall $\Delta_{\mathbf{v}} B(\mathbf{x}) \geq (1-e^{-\gamma})(1-\epsilon) \Delta_{\mathbf{v}^*}B(\mathbf{x})$ with probability at least $(1-\delta_1)(1-\delta_2) \geq 1-\delta$.
\end{proof}

% So, given $\epsilon_1 + \epsilon_2 \leq \epsilon$ and $\delta_1 + \delta_2 \leq \delta$, by generating 
% \begin{align*}
% \mathcal{N}(h,k,\epsilon,\delta) \geq \frac{\mathtt{T}^2 \mathtt{d}^{\mathtt{2\mathtt{h}}}}{\Delta_{\mathbf{x}^o}^2 D(h,\mathbf{x})} \max\Big( \frac{\ln(1/\delta_1)}{\epsilon_1^2}, \frac{\ln(\binom{n}{k}/\delta_2)}{(1-e^{-\gamma})^2\epsilon_2^2} \Big)
% \end{align*}
% assume $\mathbf{x}^g$ are our $k$-adjustment by greedy solution, we have
% \begin{align*}
% \Delta_{\mathbf{x}^g} D(h,\mathbf{x}) & \geq \Delta_{\mathbf{x}^g} \hat{D}(h,\mathbf{x}) - (1-e^{-\gamma}) \epsilon_2 \Delta_{\mathbf{x}^o} D(h,\mathbf{x}) \\
% & \geq (1-e^{-\gamma}) \Delta_{\mathbf{x}^o} \hat{D}(h,\mathbf{x}) - (1-e^{-\gamma}) \epsilon_2 \Delta_{\mathbf{x}^o} D(h,\mathbf{x}) \\
% & \geq (1-e^{-\gamma})(1-\epsilon_1) \Delta_{\mathbf{x}^o} D(h,\mathbf{x}) - (1-e^{-\gamma}) \epsilon_2 \Delta_{\mathbf{x}^o} D(h,\mathbf{x}) \\
% & \geq (1-e^{-\gamma})(1-\epsilon) \Delta_{\mathbf{x}^o} D(h,\mathbf{x})
% \end{align*}
% with probability at least $1-\delta$.

There is a drawback of the threshold (\ref{equ:first_thres}): it depends on $\Delta_{\mathbf{v}^*} B(\mathbf{x})$, which is untraceable. However, we can use the simple lower bound of $\Delta_{\mathbf{v}^*} B(\mathbf{x})$ as follows: As long as the algorithm has not terminated, there should be at least a path $p \in \mathcal{F}$ such that the length of $p$ is at most $\mathtt{T} - 1$. So the marginal gain of the optimal solution should be at least $1$. Therefore, we have the following threshold, which is the sufficient number of sample paths to bound the error between approximation ratio of $\mathbf{v}$ on $\Delta_\mathbf{v} \hat{B}(\mathcal{P}, \mathbf{x})$ and $\Delta_\mathbf{v} B(\mathbf{x})$.
\begin{align*}
\mathcal{N}(q,\epsilon,\delta) = \min_{\epsilon_1; \delta_1 } \Bigg(\mathtt{T}^2 |S|^2 \mathtt{d}^{2\mathtt{h}} \max\Big( \frac{\ln(1/\delta_1)}{\epsilon_1^2}, \frac{\ln(\binom{n+q}{q}/(\delta - \delta_1))}{2(1-e^{-\gamma})^2(\epsilon - \epsilon_1)^2} \Big)\Bigg)
\end{align*}

% \subsection{Approximation guarantee of \texttt{SA}} \label{subsec:approx_sa}
% In this part, we put together the approximation ratio of greedily selecting $\mathbf{v}$, $||\mathbf{v}|| \leq q$, on $\mathcal{P}$ and the minimum size of $\mathcal{P}$ in order to obtain the performance guarantee of \texttt{SA}. 

% Moreover, since $\Delta_{\mathbf{x}^o} D(\mathbf{x}) \geq \frac{k\gamma}{\mathtt{OPT}} (|\mathcal{F}|\mathtt{T} - D(\mathbf{x}))$. We have

% \begin{align*}
% \Delta_{\mathbf{x}^g} D(\mathbf{x}) \geq \frac{k\gamma}{\mathtt{OPT}} (1-e^{-\gamma})(1-\epsilon) (|\mathcal{F}|\mathtt{T} - D(\mathbf{x})) 
% \end{align*}

% We have enough materials to come up with the approximation ratio of Sampling method.

\begin{theorem} \label{theorem:sampling_approx}
Given $0 < \epsilon, \delta < 1$, by generating $\mathcal{N}(k,\epsilon,\frac{\delta}{||\mathtt{b}||})$ of sample paths in each sampling iteration, \texttt{SA} returns a solution within $O(\frac{\mathtt{h} \ln \mathtt{d} + \ln \mathtt{T}}{\gamma (1-e^{-\gamma})(1-\epsilon)})$ factor of optimum to the \texttt{QoSD} instance with probability at least $1-\delta$. 
\end{theorem}

\begin{skproof}
Denote $\mathbf{x}^o = \sum_i^s \mathbf{u}_i$ as an optimal solution, which is in addition to $\mathbf{x}$ to block all paths in $\mathcal{F}$ ($\mathbf{u}_i$ is a unit vector). Let $\mathbf{v}$ be a budget vector we get from greedy selection on the sample set $\mathcal{P}$. The key of our proof is that

\begin{align*}
\Delta_{\mathbf{v}} B(\mathbf{x}) \geq \frac{\gamma q}{||\mathbf{x}^o||} (1-e^{-\gamma}) (1-\epsilon) \Delta_{\mathbf{x}^o} B(\mathbf{x}) 
\end{align*}
This is proved by the finding that there exists a budget vector $\mathbf{w}$ such that $||\mathbf{w}|| \leq q$ and $\Delta_{\mathbf{w}} B(\mathbf{x}) \geq \frac{\gamma q}{||\mathbf{x}^o||} \Delta_{\mathbf{x}^o} B(\mathbf{x})$.

Therefore, we observe that: after each sampling iteration, the gap between $|\mathcal{F}|\mathtt{T}$ and $B(\mathbf{x})$ shrinks by a factor at least $(1 - \frac{k\gamma}{\mathtt{OPT}} (1-e^{-\gamma})(1-\epsilon))$ with probability at least $1 - \frac{\delta}{||\mathtt{b}||}$.

Furthermore, since there should exist at least a feasible path $p \in \mathcal{P}$ such that $\mathtt{r}(p, \mathbf{x}) \leq \mathtt{T} - 1$ before the final sampling iteration, we prove that the number of iterations is upper bounded by $O(\frac{\ln \mathtt{T} + \mathtt{h} \ln \mathtt{d}}{q \gamma (1-e^{-\gamma})(1-\epsilon)}) \mathtt{OPT}$. Since in each iteration, a budget vector $\mathbf{v}$, $||\mathbf{v}|| \leq q$, is added into solution, out final solution guarantees $O(\frac{\ln \mathtt{T} + \mathtt{h} \ln \mathtt{d}}{\gamma (1-e^{-\gamma})(1-\epsilon)})$ approximation ratio with probability at least $1-\delta$.
\end{skproof}

Interestingly, the approximation ratio of \texttt{SA} does not depend on $q$. So whatever the value of $q$ is, the result of \texttt{SA} always has the same upper bound, which means a large value of $q$ could reduce the number of sampling iterations but the number of sample paths in each iteration would increase as the trade-off.

\section{Linear Weight Functions} \label{sec:lr}

Having considered approximation algorithms to \texttt{QoSD}, we now propose a solution, called \textit{Linear Rounding} (\texttt{LR}), for the case where the edge weight functions are linear. \texttt{LR} obtains $O(\mathtt{h} \log n)$ approximation guarantee, which is the best ratio compared among all the proposed solutions.

For each $e \in E$, the weight function of $e$ is represented as $f_e(x) = \beta_e x + \alpha_e$, where $\beta_e, \alpha_e \in \mathbb{Z}^+$. Denote $\beta = \max_e \beta_e$. The \texttt{QoSD} instance can be solved by the following Integer Programming.

\begin{align}
\min & \quad \sum_{e \in E} x_e & \label{eqn:IP} \\
\text{s.t. } & \quad \sum_{e \in p} (\beta_e x_e + \alpha_e) \geq \mathtt{T} && \forall  p \in \mathcal{F} \label{con:path} \\
& \quad x_e \leq b_e && \forall e \in E \label{con:max_cost}\\
& \quad x_e \in \mathbb{Z}^+ \cup \{0\} && \forall e \in E \label{con:qos_integer}
\end{align}
This IP has a simple linear relaxation by replacing constraint (\ref{con:qos_integer}) with:
\begin{align}
x_e \in \mathbb{R}^+ \cup \{0\} && \forall e \in E
\end{align}

% In this section, we will consider a special case of QoS cut problem where the delay function of a edge is linear with the QoS level of this edges $f_e(x) = \beta_e x + \alpha_e$ for all $e \in E$. So the QoS cut problem can be solved by following Integer Programming:

% \begin{align}
% \min & \sum_{e \in E} x_e & \label{eqn:IP} \\
% \text{s.t. } & \sum_{e \in p} (\beta_e x_e + \alpha_e) \geq \mathtt{T} && \forall  p \in \mathcal{F} \label{con:path} \\
% & x_e \in \mathbb{Z}^+ \cup \{0\} && \forall e \in E \label{con:qos_integer}
% \end{align}
% where $x_e$ is the QoS level of edge $e$. The constraint (\ref{con:path}) guarantee every possible path connecting transactions should have the overall delay no smaller than $\mathtt{T}$. This IP has a simple linear relaxation by replacing constraint (\ref{con:qos_integer}) with:

% \begin{align}
% x_e \in \mathbb{R}^+ \cup \{0\} && \forall e \in E
% \end{align}

Although constructing this relaxation maybe intractable due to the extremely large size of $\mathcal{F}$, this LP still can be solved in polynomial time using \textit{ellipsoid method} with a simple separation oracle similar to Multicut problem \cite{vazirani2013approximation}. 

Denote the vector $\mathbf{x}^\prime = \{x_1^\prime, ...x_m^\prime\}$ as the optimal solution to the LP relaxation, $x_e^\prime$ can be a real number. The problem now is how to obtain a discrete solution $\mathbf{x}$ from $\mathbf{x}^\prime$ and what approximation guarantee $\mathbf{x}$ provides? To do so, we applied the randomized rounding technique as follows: Given an edge $e$, if $x_e^\prime$ is an integer, let $x_e = x_e^\prime$. Otherwise, denote $\rho_e = x_e^\prime - \floor{x_e^\prime}$ and given $\eta$, which would be defined later, then:
\begin{itemize}
\item If $\eta \rho_e \geq 1$, $x_e = \ceil{x_e^\prime}$. Let $y_e = x_e$
\item If $\eta \rho_e < 1$, $x_e = \ceil{x_e^\prime}$ with probability $\eta \rho_e$ and $\floor{x_e^\prime}$ otherwise. Let $y_e = \floor{x_e^\prime}$
\end{itemize}
\texttt{LR} is fully presented in Alg. \ref{alg:linear_rounding}.
\begin{algorithm}[t]
	\caption{Linear Rounding algorithm (\texttt{LR})}
    \label{alg:linear_rounding}
	\begin{flushleft}
	\textbf{Input} $G, S, \mathtt{T}, \mathtt{b}, f, \delta$\\
	\textbf{Output} cost vector $\mathbf{x} = \{x_1,..x_m\}$
	\end{flushleft}
    \begin{algorithmic}[1]
        \State $\mathbf{x}^\prime \leftarrow $ optimal solution of LP-relaxation.
        \State $\beta = \max_e \beta_e$; $\eta = \frac{\beta}{1-\exp(-\beta)} (\ln \frac{n^\mathtt{h}}{\delta} + 1)$
        \For{each $e \in E$}
        	\If{$x_e^\prime$ is a integer}
            	\State $x_e = x_e^\prime$
            \Else
            	\State $\rho_e = x_e^\prime - \floor{x_e^\prime}$
                \State $x_e = \ceil{x_e^\prime}$ with probability $\eta \rho_e$; $\floor{x_e^\prime}$ otherwise.
            \EndIf
        \EndFor
     \end{algorithmic}
     \begin{flushleft}
     	\textbf{Return } $\mathbf{x}$
     \end{flushleft}
\end{algorithm}

Consider a path $p \in \mathcal{F}$, it is trivial that $\mathbf{x}$ will block $p$ if $\sum_{e \in p} f_e(y_e) \geq \mathtt{T}$. The question is whether $\mathbf{x}$ can block $p$ if $\sum_{e \in p} f_e(y_e) < \mathtt{T}$? Denote:
\begin{align*}
T_p = \mathtt{T} - \sum_{e \in p} f_e(y_e) \\
\mathcal{E}_p = \{e \in p; y_e < x_e\}
\end{align*}
So:
\begin{align*}
T_p \leq \sum_{e \in \mathcal{E}_p} \beta_e (x^\prime_e - y_e) = \sum_{e \in \mathcal{E}_p} \beta_e \rho_e
\end{align*}
Then the probability that $\mathbf{x}$ does not block $p$ is given as follows:

% $x_e = \ceil{x_e^\prime}$ with probability $\eta \rho_e$ and $\floor{x_e^\prime}$ otherwise. It is trivial that $x_e$ satisfies the constraint (\ref{con:max_cost}) for all $e \in E$. Since there is probability that $x_e < x_e^\prime$, the natural question is whether $\mathbf{x}$ can block all paths in $\mathcal{F}$? Consider a path $p \in \mathcal{F}$, denote $T_p = \mathtt{T} - \sum_{e \in p} (\beta_e \floor{x_e^\prime} + \alpha_e)$. We have $T_p \leq \sum_{e \in p} \beta_e (x_e^\prime - \floor{x_e^\prime}) = \sum_{e \in p} \beta_e \rho_e$ due to constraint (\ref{con:path}). If there exists $e \in p$ such that $x_e^\prime$ is real, then $T_p \geq 1$. The probability that $\mathbf{x}$ does not block $p$ is given as follows:

\begin{align}
&\textnormal{Pr}[p \textnormal{ is not blocked by } \mathbf{x}] = \textnormal{Pr}\Big[\sum_{e \in p} (\beta_e x_e + \alpha_e) < \mathtt{T} \Big] = \textnormal{Pr}\Big[ \sum_{e \in \mathcal{E}_p} \beta_e (x_e - \floor{x_e^\prime}) < T_p \Big] \\
& \quad \quad = \textnormal{Pr}\Big[\exp{\Big(-\sum_{e \in \mathcal{E}_p} \beta_e (x_e - \floor{x_e^\prime}) \Big) } > \exp{(-T_p)}\Big] \\
& \quad \quad \leq \exp(T_p) \cdot \mathbb{E}\Big[\exp\Big(-\sum_{e \in \mathcal{E}_p} \beta_e (x_e - \floor{x_e^\prime}) \Big) \Big] \label{equ:markov} \\
% & \quad \quad = \exp(T_p) \cdot \prod_{e \in \mathcal{E}_p} \mathbb{E}\Big[\exp \Big(-\beta_e (x_e - \floor{x_e^\prime}) \Big)\Big] \\
& \quad \quad = \exp(T_p) \cdot \prod_{e \in \mathcal{E}_p} \Big( \exp(-\beta_e) \cdot \eta\rho_e + (1 - \eta\rho_e) \Big) \\
& \quad \quad \leq \exp(T_p) \cdot \Bigg( 1 - \frac{\sum_{e \in \mathcal{E}_p} \eta \rho_e (1 - \exp(-\beta_e))}{|\mathcal{E}_p|} \Bigg)^{|\mathcal{E}_p|} \label{equ:cauchy} \\
& \quad \quad \leq \exp(T_p) \cdot \Bigg( 1 - \frac{1-\exp(-\beta)}{\beta} \cdot \frac{\eta T_p}{|\mathcal{E}_p|} \Bigg)^{|\mathcal{E}_p|} \\
& \quad \quad \leq \exp(T_p) \cdot \exp\Bigg(- \eta T_p \frac{1 - \exp(-\beta)}{\beta} \Bigg) \leq \exp\Bigg( -\bigg( \eta \frac{1-\exp(-\beta)}{\beta} - 1 \bigg) \Bigg)
% & \quad \quad = \exp(T_p) \cdot \prod_{e \in p} \Big(\eta \cdot \rho_e \cdot e^{-\beta_e \ceil{x_e^\prime} - \alpha_e} + (1 - \eta \cdot \rho_e) e^{-\beta_e \floor{x_e^\prime} - \alpha_e} \Big) \\
% & \quad \quad \leq e^{\mathtt{T} - \alpha_{min}} \Big( 1 - \frac{\sum_{e \in p} kx_e^{LP}(1-e^{-\beta_e})}{|p|} \Big)^{|p|} \label{equ:cauchy} \\
% & \quad \quad \leq e^{\mathtt{T} - \alpha_{min}} \Big( 1 - \frac{1-e^{-\beta_{max}}}{\beta_{max}} \frac{k\mathtt{T}}{|p|}  \Big)^{|p|} \label{equ:mon} \\
% & \quad \quad \leq e^{\mathtt{T} - \alpha_{min}} e^{-\frac{1-e^{-\beta_{max}}}{\beta_{max}} k\mathtt{T}} = e^{-\mathtt{T} (k\frac{1-e^{-\beta_{max}}}{\beta_{max}} - 1) - \alpha_{min}}
\end{align}
Eq. \ref{equ:markov} comes from Markov inequality \cite{markovInequality} while Eq. \ref{equ:cauchy} is from Cauchy Theorem \cite{cauchyInequality}. Since there are at most $n^\mathtt{h}$ feasible paths in $\mathcal{F}$, using Union Bound theory \cite{UnionBound}, the probability that $\mathbf{x}$ cannot block all paths in $\mathcal{F}$ is at most
\begin{align} \label{equ:prob_x}
n^\mathtt{h} \cdot \exp\Bigg( -\bigg( \eta \frac{1-\exp(-\beta)}{\beta} - 1 \bigg) \Bigg)
\end{align}
\begin{theorem}
Given fixed $0 < \delta < 1$ and $\eta = \frac{\beta}{1-\exp(-\beta)} (\ln \frac{n^\mathtt{h}}{\delta} + 1)$, \texttt{LR} returns a solution within $O(\mathtt{h} \ln n)$ factor of optimum to the \texttt{QoSD} instance with probability at least $1-\delta$.
\end{theorem}
\begin{proof}
From Eq. \ref{equ:prob_x} and the given $\eta$, the probability that $\mathbf{x}$ blocks all paths in $\mathcal{F}$ is at least $1-\delta$. Also
\begin{align*}
&\mathbb{E}[||\mathbf{x}||] = \sum_e \Big( \ceil{x_e^\prime} \eta \rho_e + \floor{x_e^\prime} (1 - \eta \rho_e) \Big) \\
& \quad \leq \sum_e \Big( \floor{x_e^\prime} + \eta \rho_e \Big) \leq \sum_e \eta x_e^\prime \\
& \quad = \eta \cdot ||\mathbf{x}^\prime|| \leq \eta \cdot ||\mathbf{x}^*||
\end{align*}
which completes the proof.
\end{proof}

\section{Discussion} \label{sec:discussion}
In this section, we discuss the trade-off between the performance guarantee and the runtime complexity of the four proposed algorithms, summarized in Table. \ref{table:performance}.

First, we consider the performance guarantee of the \texttt{IG} and \texttt{AT} algorithm. The approximation ratio of \texttt{IG} and \texttt{AT} are $O(\frac{1}{\gamma} (\mathtt{h} \ln n + \ln \mathtt{T}) )$ and $O(\mathtt{h} \ln n + \ln \mathtt{T})$ respectively, where $\gamma$ is the \textit{concave ratio} of edge weight functions. $\gamma$ plays an important role in the differences between \texttt{IG} and \texttt{AT} solutions. The smaller $\gamma$ is - which signifies a more convex of edge weight functions - the worse \texttt{IG} performs. But if all edge weight functions are concave - or at least linear - $\gamma$ equals to $1$, then \texttt{IG} and \texttt{AT} obtain the same approximation guarantee. Not only achieve the same ratio, the two algorithms also return the same solution because in \texttt{AT}, $\frac{\Delta_{u(e,x)} \mathtt{D}(\mathcal{P}, \mathbf{x})}{x}$ reaches maximum at $x = 1$. So in each iteration, the budget increases at most by $1$, and it is also the selection of \texttt{IG}. Overall, \texttt{AT} theoretically returns better solutions than \texttt{IG}. 

%  The maximum number of outer iterations (line \ref{line:outer_it} of Alg. \ref{alg:iterative}) is upper bounded by $n^\mathtt{h}$, which is theoretically a huge number. However, from all our experiment on both random graphs or real-but-dense networks, the number of iterations never reach this amount.

% in each inner iteration (line \ref{line:inner_it_trunk} of Alg. \ref{alg:trunk_adding} and line \ref{line:inner_it_greedy} of Alg. \ref{alg:greedy_blocking})

However, in trade-off, \texttt{AT} has higher computational complexity than \texttt{IG}. Both algorithms use the same framework as in Alg. \ref{alg:iterative}. The maximum number of  iterations in this framework (line \ref{line:outer_it} of Alg. \ref{alg:iterative}) is upper bounded by $n^\mathtt{h}$, which is theoretically a large number. However, from our experiments on both random graphs and real-but-dense networks, the number of iterations never reach this amount. Considering the strategy of blocking paths, the number of computation in each inner iteration (line \ref{line:inner_it_trunk} of Alg. \ref{alg:trunk_adding}) of \texttt{AT} is $O(||\mathtt{b}||)$, while this number (line \ref{line:inner_it_greedy} of Alg. \ref{alg:greedy_blocking}) in \texttt{IG} is $O(m)$. In the worst-case scenario, the number of inner iterations of both \texttt{IG} and \texttt{AT} can reach up to $O(||\mathtt{b}||)$. Therefore, the worst-case runtime complexity of \texttt{IG} and \texttt{AT} is $O(n^\mathtt{h} m ||\mathtt{b}||)$ and $O(n^\mathtt{h} ||\mathtt{b}||^2)$ respectively.

With \texttt{SA}, to obtain the $O(\frac{\ln \mathtt{T} + \mathtt{h} \ln d}{\gamma (1-e^{-\gamma}) (1-\epsilon)})$ ratio, we have to generate $\mathcal{N}(k,\epsilon,\delta/||\mathtt{b}||)$ paths with $O(m)$ time complexity for each path in each sampling steps. Also, after sampling, a budget vector $\mathbf{v}$ ($||\mathbf{v}|| \leq q$) is added into $\mathbf{x}$, which makes the number of sampling steps at most $O(\frac{||\mathtt{b}||}{q})$. Moreover, the greedy selection on a sample set costs $O(qm)$ runtime complexity. Therefore, the worse-case runtime complexity of \texttt{SA} is bounded by $O(\mathcal{N}(k,\epsilon, \frac{\delta}{||\mathtt{b}||}) \cdot ||\mathtt{b}|| \cdot m^2)$. However, if $\mathtt{h}$ is large, the number of samples required by \texttt{SA} becomes large and its sampling procedure dominates its runtime; this is ameliorated by trivially parallelizing the sampling process, which is possible since each sample is independent. In practice, the parameter $\alpha$ greatly reduces the required number of samples; with $\alpha = 0.8$, we found that $O(|S|)$ samples were sufficient to provide feasible solutions within reasonable runtime. 

% Note that generating $\mathcal{N}(k,\epsilon_1, \epsilon_2,\delta_1, \delta_2)$ is a huge number and a burden to memory and running time as well. However, different to  Iterative algorithm, which - theoretically in worst case - has to store $O(n^\mathtt{h})$ paths to obtain a feasible solution, we could reduce the number of sampling paths of the Sampling algorithm while guaranteeing to acquire a feasible solution. This reduction makes the approximation guarantee of Sampling algorithm no more effective. In our experiment, we observe that if the bias ratio $\alpha$ on Alg. \ref{alg:sampling_path} is high enough, only $O(|S|)$ samples were sufficient to obtain a good solution. $O(|S|)$ is clearly much more smaller than $\mathcal{N}(k,\epsilon,\delta)$, which makes the Sampling algorithm to be the algorithm that consume least memory among our proposed solutions. Also, with the appearance of parameter $k$, which boosts up the progress of reaching a feasible solution, Sampling algorithm performs very fast within sparse networks, whose number of feasible paths is small. Therefore, Sampling algorithm can be considered as a good choice on several cases that we need a fast solution which could scale up with the changes of network.

\begin{table}[t]
\centering
\caption{Algorithm performance ratio and time complexity} \label{table:ratio}
\label{table:performance}
\begin{tabular}{c | l | l }
\toprule
Algorithm     & Approximation Ratio    & Worst-case Runtime \\ 
 \midrule
\texttt{IG}	 			& $O(\gamma^{-1} (\ln \mathtt{T} + \mathtt{h} \ln n))$ 	& $O(n^\mathtt{h} m ||\mathtt{b}||)$ \\
\texttt{AT}  		& $O(\ln \mathtt{T} + \mathtt{h} \ln n)$ 		& $O(n^\mathtt{h} ||\mathtt{b}||^2)$\\
\texttt{SA}	  		& $O(\frac{\ln \mathtt{T} + \mathtt{h} \ln \mathtt{d}}{\gamma (1-e^{-\gamma}) (1-\epsilon)})$ 		& $\mathcal{N}(q,\epsilon,\frac{\delta}{||\mathtt{b}||}) \cdot O(m^2||\mathtt{b}||)$  \\
\texttt{LR}		& $O(\mathtt{h} \ln n)$ 	& LP-solver$() + m$ \\
\bottomrule
\end{tabular}
\end{table}

Next, consider the \texttt{LR} solution, which is only used if all the weight functions are linear. The runtime of \texttt{LR} strongly depends on the linear programming solver (LP-solver$()$). In the experimental evaluation, we observe that in most cases, the number of edges - whose $x_e^\prime$ is real - is inconsiderably small. Therefore, after randomized rounding, the size of the discrete solution $\mathbf{x}$ has a diminutive gap comparing with $\mathbf{x}^\prime$'s. Hence, although \texttt{IG} and \texttt{AT} perform fairly well in general cases, \texttt{LR} usually returns the best solution if the edge weight functions are linear.

\section{Experimental Evaluation} \label{sec:experiment}
% \mthai{fill in Kunle part}
In this section, we evaluate our proposed approximation algorithms by 1) comparing their performance to an intuitive heuristic as there is no other solution to \texttt{QoSD}, in a general case; and 2) comparing our algorithms to \cite{kuhnle2018network} as a special case of \texttt{QoSD}. The experiments were conducted on a Linux machine with 2.3Ghz Xeon 18 core processor and 256GB of RAM. The programming language we used is C++. Several steps in our algorithm are parallelized by using OpenMP with 64 threads. The reported running time is real-world time, not CPU time. The source code is available at \cite{code}.

\subsection{Experiment Settings}
We evaluated the following algorithms; the source code of all of our implementation is written in C++.
\begin{itemize}
\item \texttt{AT}: In this solution, to find the shortest paths between a pair of nodes, we utilized the Dijkstra algorithm and computed each path separately. The reason for this implementation is that by doing so, we can parallelize the process by dividing it into independent tasks. Therefore, even the theoretical time complexity of the Dijkstra algorithm for all-pair shortest paths is worse than Floyd-Warshall methods, the parallelization helps to boost the performance of the Dijkstra algorithm while it is impossible to do so with Floyd-Warshall. 
\item \texttt{IG}: this algorithm used the same settings as \texttt{AT}.
\item \texttt{SA}: We set the bias ratio $\alpha = 0.8$ and the number of sample paths is $O(|S|)$ for all experiments. We found this value of $\alpha$ and the number of samples are sufficient to obtain feasible solutions within reasonable runtime in most cases. 
\item \texttt{LR}: We used CPLEX \cite{cplex2009v12} to solve the linear programming. Implementing the \textit{ellipsoid method} could result in impractical performance. So we used the same concept of \texttt{IG} and \texttt{AT} to solve the LP relaxation as follows: rather than listing all feasible paths in $\mathcal{F}$, we iteratively listed the shortest LP-weighted paths as constraints until the length of the shortest paths between each pair exceeded $\mathtt{T}$.  
\item Centrality Cutting ($\mathtt{CC}$) heuristic: Centrality has been commonly used as a metric to identify critical components of a network in the literature. \texttt{CC} works in iterative manner as follows: First, we set $\mathbf{x} = \{0\}^m$. In each iteration, we found the shortest paths between a pair of nodes under the current budget vector $\mathbf{x}$ and computed the number of appearances of each edge in those paths. The algorithm then raised the weight of the edge that appears the most to maximum. All those steps are repeated until there were no shortest paths whose length was smaller than $\mathtt{T}$. When finding shortest paths of each pair, we also used parallelization to boost $\mathtt{CC}$ performance.
\item \texttt{SAP}, \texttt{MIA}, \texttt{TAG} \cite{kuhnle2018network}: These algorithms were only implemented in comparison on the special case of \texttt{QoSD} (the \texttt{LB-MULTICUT} problem). The source code of those algorithms was taken from \cite{lbcode} and it was only available for undirected networks.
\end{itemize}

To obtain $S$, we sampled uniformly random sets of pairs of nodes on each network. All results were averaged over 5 independent repetitions of each experiment. The weight function of each edge was selected from following functions

\begin{itemize}
\item A linear function $f_e(x) = \Theta(x)$. 
\item A convex function $f_e(x) = \Theta(x^2)$. This function was inspired by the average delay calculation on computer networks w.r.t packet arrival rate.
\item A concave function $f_e(x) = \Theta(\ln x)$. This function was inspired by the additive metric on IoT network w.r.t packet error rate. 
\item A cutting function: $f_e$ only received two values, $f_e(0) = 1$ and $f_e(1) = \mathtt{T}$. This function was used when we compared our solution with the algorithms of the \texttt{LB-MULTICUT} problem.
\end{itemize}
Each function was set such that the initial weight $f_e(0)=1$ and the maximum weight was $\mathtt{T}$. Since there exists heterogeneous coupling delays in modern networks, in our experiment, the weight function of each edge was randomly selected from the linear, convex or concave functions as mentioned above. In the experiments with the presence of \texttt{LR}, all weight functions were linear. On the other hand, all weight functions were cutting function if compared with the algorithms of \texttt{LB-MULTICUT}. 

% Each experiment will be implemented in five settings of delay function: all functions are (1) linear, (2) concave, (3) convex, (4) mixed up with those three functions: each edge's weight function is randomly selected among those three; and (5) cutting functions. 
% \begin{itemize}
% \item A linear function $f_e(x) = O(x)$, $b_e$ is set to be $\mathtt{T} - 1$. 
% \item A convex function $f_e(x) = O(x^2) + 1$ and $b_e = \ceil{\sqrt[]{\mathtt{T} - 1}}$. This function is inspired by average delay calculation on routers w.r.t packet arrival rate.
% \item A concave function $f_e(x) = \ln(Bx + 1) + 1$ where $B$ is adjustable parameter to ensure that $b_e = 10$ and $ f_e(b_e) \leq \mathtt{T}$. This function is inspired by the additive metric on IoT network w.r.t packet error rate. 
% \item A cutting function: $f_e$ only receives two values, $f_e(0) = 1$ and $f_e(1) = \mathtt{T}$. This function is used when we compare our solution with the algorithms of the \texttt{LB-MULTICUT} problem.
% \end{itemize}

The algorithms were implemented on both synthesized networks and real-world networks. The synthesized networks we used were the Erdos-Renyi (ER) \cite{erdos1960evolution} graphs with $240$ nodes and varied the edge density parameter $\rho$. For the real-world networks, we used the datasets from Stanford Network Analysis Project \cite{snapnets}, including Gnutella, Skitter and Roadnet. Skitter is highly dense IPv4 Internet topology graph, which were collected by traceroutes run daily in 2005; Gnutella is the snapshots of peer-to-peer file sharing; and RoadCA is a road network of California where intersections and endpoints are represented by nodes, and the roads connecting these intersection or endpoints are represented by undirected edges. Information of real-world datasets are summarized in Table \ref{table:dataset}. 
% In case the input is undirected, we treat each undirected edge $(u,v)$ as two directed edge $(u,v)$ and $(v,u)$ but whenever the cost on one edge is adjusted, the cost of the other one is changed as well.

\begin{table}[t]
\centering
\caption{Statistics of datasets} \label{table:dataset}
\begin{tabular}{l | l | r | r | r}
\toprule
     Data & Type    & Nodes & Edges & Diameter \\ 
 \midrule
Gnutella & Directed 	& $10.9$ K & $40.0$ K  & 9 \\
% Bitcoin  & Undirected 	&    	& 	\\
RoadCA & Undirected 	& $2.0$ M 	& $2.8$ M & 786\\
Skitter  & Directed 	& $1.7$ M 	& $11.1$ M & 25 \\
\bottomrule
\end{tabular}
\end{table}

\subsection{Performance comparison}

\subsubsection{Small size random graph}
\begin{figure}[t]
\subfloat[Linear]{
	\includegraphics[width=0.35\textwidth]{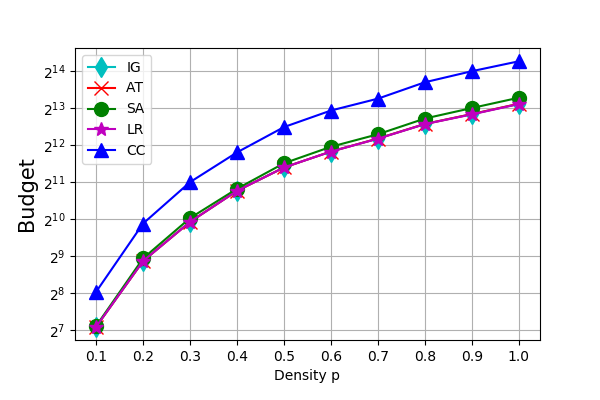}
	\label{fig:randDelay1_size}}
~
\hspace{-20px}
 \subfloat[Heterogeneous]{\includegraphics[width=0.35\textwidth]{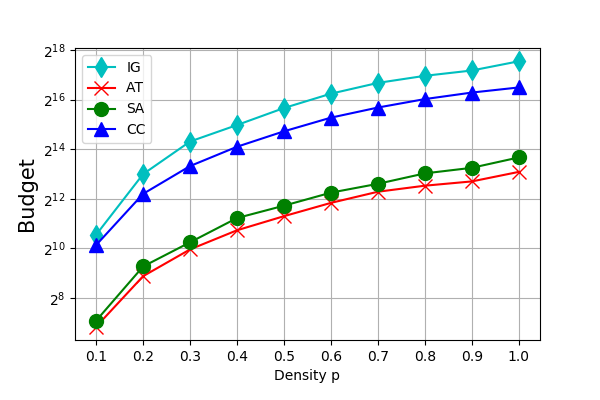}
\label{fig:randDelay2_size}}
~
\hspace{-20px}
\subfloat[Cutting]{\includegraphics[width=0.35\textwidth]{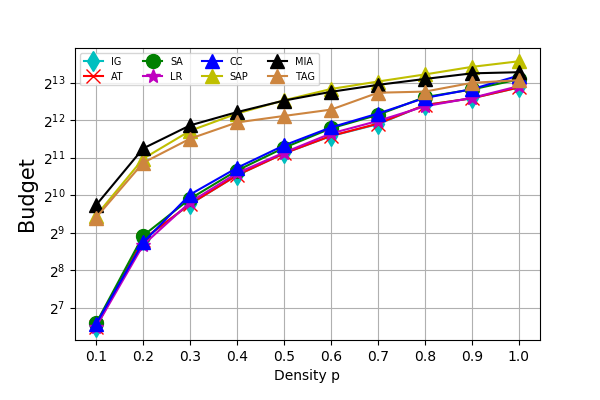}
\label{fig:randDelay5_size}}
    
    \caption{Solution quality of algorithms on Random graph}
 	\label{fig:random_size}
\end{figure}

\begin{figure}[t]
\subfloat[Linear]{
	\includegraphics[width=0.35\textwidth]{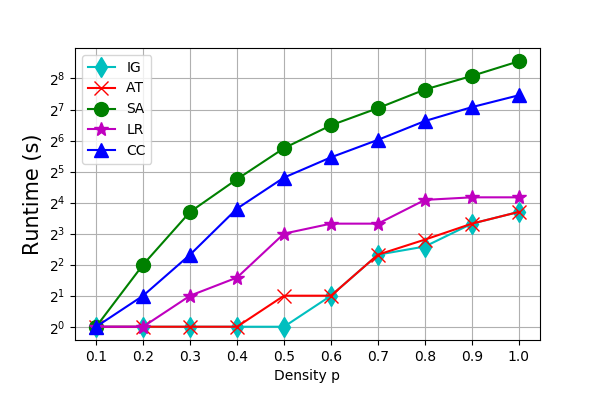}
	\label{fig:randDelay1_time}}
~
\hspace{-20px}
 \subfloat[Heterogeneous]{\includegraphics[width=0.35\textwidth]{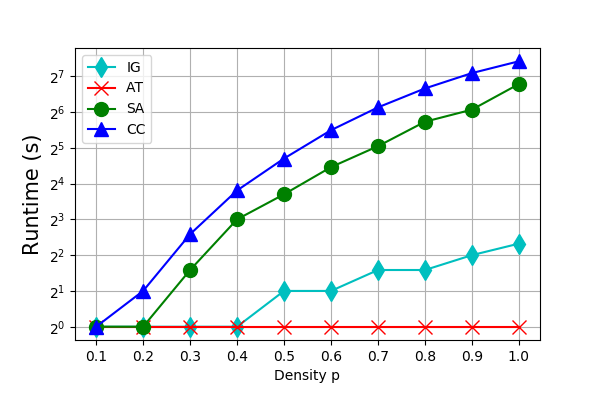}
\label{fig:randDelay2_time}}
~
\hspace{-20px}
\subfloat[Cutting]{\includegraphics[width=0.35\textwidth]{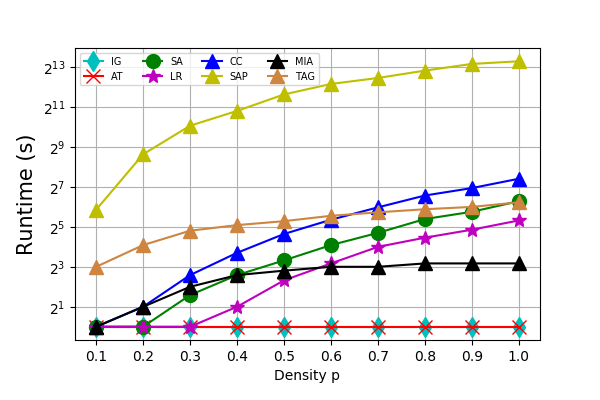}
\label{fig:randDelay5_time}}
    
    \caption{Runtime of algorithms on Random graph}
 	\label{fig:random_time}
\end{figure}

% \begin{figure}[t]
% \subfloat[Solution size]{
% 	\includegraphics[width=0.25\textwidth]{image/random_directed_delay1_changeP_result.png}
% 	\label{fig:randDelay1_size}}
% ~
%  \subfloat[Running time]{
%   	\includegraphics[width=0.25\textwidth]{image/random_directed_delay1_changeP_time.png}
%   	\label{fig:randDelay1_time}}
    
%     \caption{Random graph with linear weight functions}
%  	\label{fig:randDelay1}
% \end{figure}

% \begin{figure}[t] 
% \subfloat[Solution size]{\includegraphics[width=0.25\textwidth]{image/random_directed_delay2_changeP_result.png}
% \label{fig:randDelay2_size}}
% ~
%  \subfloat[Running time]{
%   	\includegraphics[width=0.25\textwidth]{image/random_directed_delay2_changeP_time.png}
%   	\label{fig:randDelay2_time}}
    
%     \caption{Random graph with heterogeneous weight functions}
%  	\label{fig:randDelay2}
% \end{figure}

% \begin{figure}[t] 
% \subfloat[Solution size]{\includegraphics[width=0.25\textwidth]{image/random_undirected_delay5_changeP_result.png}
% \label{fig:randDelay5_size}}
% ~
%  \subfloat[Running time]{
%   	\includegraphics[width=0.25\textwidth]{image/random_undirected_delay5_changeP_time.png}
%   	\label{fig:randDelay5_time}}
    
%     \caption{Random graph with cutting weight functions}
%  	\label{fig:randDelay5}
% \end{figure}

In these experiments, we compared our algorithms with the $\mathtt{CC}$ solution on directed ER networks with $n=240$ and we varied the edges density $\rho$. The threshold \texttt{T} was set to be 3 and the size of $S$ was 10.

% Both solution quality and runtime are shown in  Fig.\ref{fig:randDelay1}. 

Fig. (\ref{fig:randDelay1_size}) and Fig. (\ref{fig:randDelay1_time}) show the results and runtimes of the algorithms when edge weight functions are linear. We notice that our four algorithms performed almost similarly in terms of quality of solution and very close to the optimal solution of LP relaxation. Meanwhile, $\mathtt{CC}$ was far from being optimal when its solutions were always at least double to the solution of other algorithms. In terms of runtime, the ranking from best to worst was $\mathtt{IG}$, $\mathtt{AT}$, $\mathtt{LR}$, $\mathtt{CC}$ and $\mathtt{SA}$. $\mathtt{SA}$ performed worse especially when the edge density increased and approached to 1. This can be explained by the following: with high value of the bias parameter $\alpha$, most paths of the sample set $\mathcal{P}$ were the shortest paths of pairs in $S$. But because the edge density is high, there would be multiple paths between a pair of nodes whose length is smaller than $\mathtt{T}$, which makes the number of sampling iterations on $\mathtt{SA}$ increases. Therefore, $\mathtt{SA}$ had a high runtime on finding shortest paths and then sampling, which was the main factor degrading its runtime. We observe that at the smallest edge density ($0.1$), all algorithms performed the best on both quality of solution and runtime, which is promising since most of real-world networks are sparse \cite{kuhnle2018network}.

Next, we compared our algorithms in the scenario with heterogeneous weight functions. $\mathtt{LR}$ was no longer applicable, which explains why we did not plot $\mathtt{LR}$ in Fig. (\ref{fig:randDelay2_size}) and Fig. (\ref{fig:randDelay2_time}). Although having the same quality of solution in linear delay, $\mathtt{IG}$ performed much worse than $\mathtt{AT}$ when its sizes of solutions were always at least 20 times of $\mathtt{AT}$'s. The concave ratio $\gamma$ was 0 in this scenario because there existed a weight function whose value did not change by adding several cost units. This experiment clearly illustrated the impact of concave ratio $\gamma$ on the performance guarantee of $\mathtt{IG}$. Moreover, $\gamma$ also impacted on $\mathtt{IG}$'s runtime because the \texttt{IG} runtime is proportional to the solution size. From density $0.4$, the gap between $\mathtt{IG}$ and $\mathtt{AT}$'s runtime became distinguishable. 

% Despite this result, $\mathtt{IG}$ is still not a bad algorithm, we will demonstrate in next subsection that $\mathtt{IG}$ performs well on real-world networks with various types of weight functions.

Finally, we compared our algorithms with three methods proposed by Kuhnle et al. \cite{kuhnle2018network} for the \texttt{LB-MULTICUT} problem. The random graph was undirected and had 240 nodes. The size of solution and runtimes are reported in Fig. (\ref{fig:randDelay5_size}) and Fig. (\ref{fig:randDelay5_time}). All of our four algorithms returned the best results while \texttt{SAP}, \texttt{MIA} and \texttt{TAG} were even worse than \texttt{CC} in term of quality of solution. The gap between the final budgets of those three algorithms and our algorithms was significant with small $\rho$ and became smaller when $\rho$ increased. In terms of runtime, \texttt{SAP} and \texttt{TAG} performed the worst while \texttt{MIA} bypassed \texttt{SA} after $\rho = 0.3$ and \texttt{LR} after $\rho = 0.5$. \texttt{IG} and \texttt{AT} were the fastest by far. It took less than one second for these two algorithms to finish no matter the edge density.

\subsubsection{Results on real networks}

In this subsection, we evaluate our algorithms on the real-world networks. We mainly examined the effect of varying the threshold $\mathtt{T}$ on the algorithm performances. The number of pairs is set to be 100. We limited the runtime by a day (24 hours); any experiments, which ran longer than a day, were terminated.

\begin{figure}[t] 
\subfloat[Gnutella - Linear]{\includegraphics[width=0.35\textwidth]{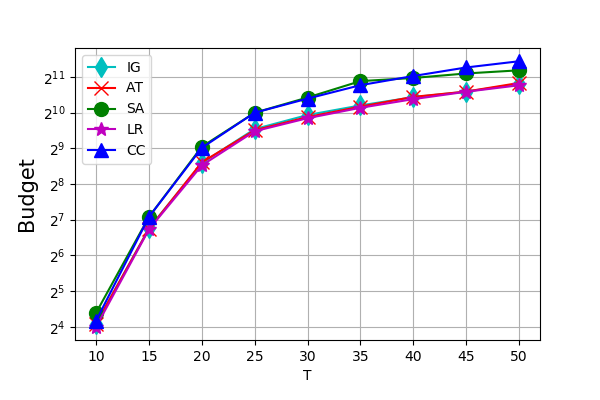}
\label{fig:gnutellaDelay1_size}}
~
\hspace{-20px}
\subfloat[Gnutella - Heterogeneous]{\includegraphics[width=0.35\textwidth]{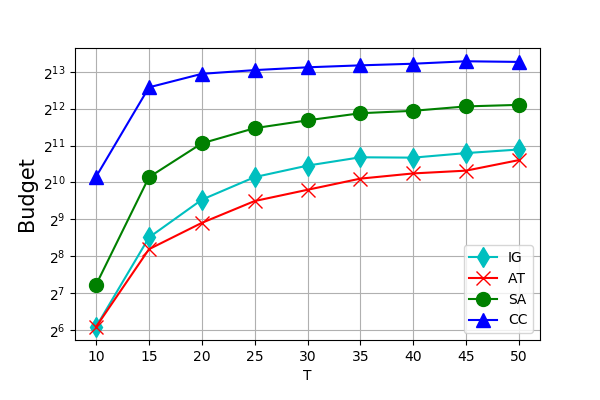}
\label{fig:gnutellaDelay4_size}}
~
\hspace{-20px}
\subfloat[Skitter - Heterogeneous]{\includegraphics[width=0.35\textwidth]{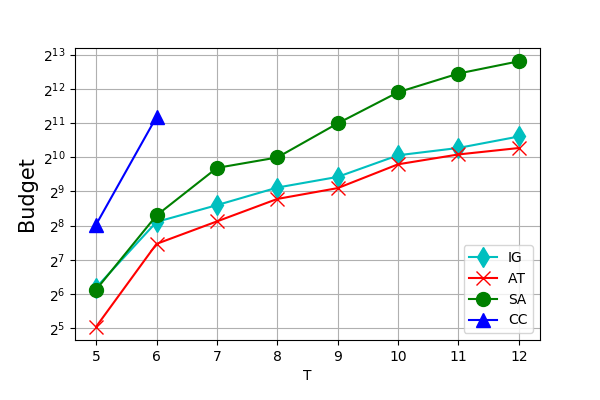}
\label{fig:skitterDelay4_size}}
    
    \caption{Solution quality of algorithms on Gnutella and Skitter}
 	\label{fig:gnuski_size}
\end{figure}

\begin{figure}[t] 
\subfloat[Gnutella - Linear]{\includegraphics[width=0.35\textwidth]{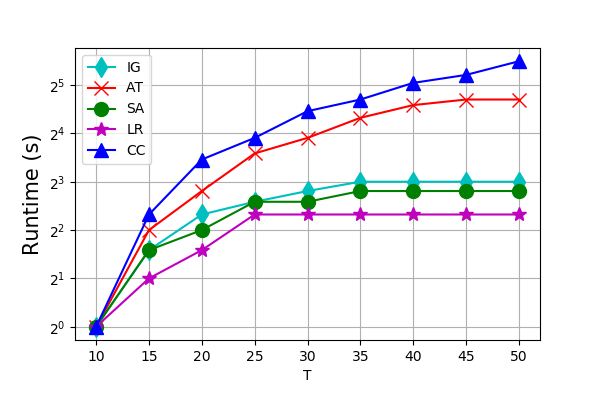}
\label{fig:gnutellaDelay1_time}}
~
\hspace{-20px}
\subfloat[Gnutella - Heterogeneous]{\includegraphics[width=0.35\textwidth]{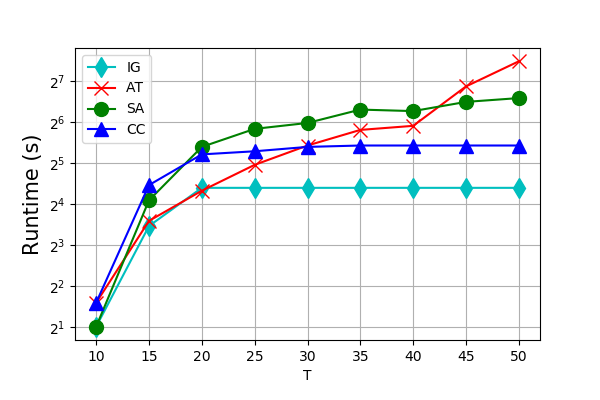}
\label{fig:gnutellaDelay4_time}}
~
\hspace{-20px}
\subfloat[Skitter - Heterogeneous]{\includegraphics[width=0.35\textwidth]{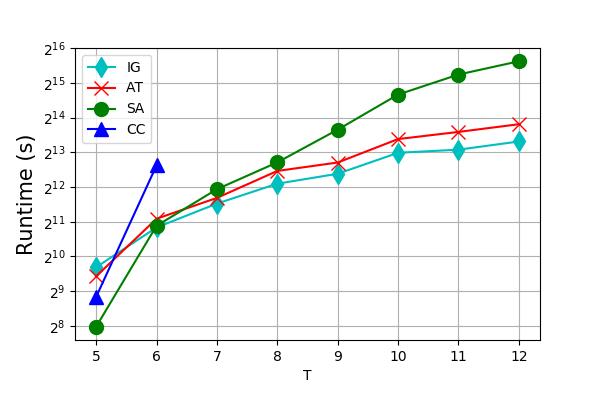}
\label{fig:skitterDelay4_time}}
    
    \caption{Runtime of algorithms on Gnutella and Skitter}
 	\label{fig:gnuski_time}
\end{figure}

% \begin{figure}[t] 
% \subfloat[Solution size]{\includegraphics[width=0.5\textwidth]{image/p2p-Gnutella04_directed_delay1_changeT_result.png}
% \label{fig:gnutellaDelay1_size}}
% ~
% % \hspace{-20px}
%  \subfloat[Running time]{
%   	\includegraphics[width=0.5\textwidth]{image/p2p-Gnutella04_directed_delay1_changeT_time.png}
%   	\label{fig:gnutellaDelay1_time}}
    
%     \caption{Gnutella with linear weight functions}
%  	\label{fig:gnutellaDelay1}
% \end{figure}

% \begin{figure}[t] 
% \vspace*{-20px}
% \subfloat[Solution size]{\includegraphics[width=0.25\textwidth]{image/p2p-Gnutella04_directed_delay4_changeT_result.png}
% \label{fig:gnutellaDelay4_size}}
% ~
% \hspace{-20px}
%  \subfloat[Running time]{
%   	\includegraphics[width=0.25\textwidth]{image/p2p-Gnutella04_directed_delay4_changeT_time.png}
%   	\label{fig:gnutellaDelay4_time}}
    
%     \caption{Gnutella with heterogeneous weight functions}
%  	\label{fig:gnutellaDelay4}
% \end{figure}

% In this subsection, we evaluate our algorithms on the real-world networks. We mainly examined the effect of varying the threshold $\mathtt{T}$ on the algorithm performances. The number of pairs is set to be 100. We limited the runtime by a day (24 hours); any experiments, which ran longer than a day, were terminated.

First, we discuss the results on the smallest network, Gnutella, in which we let $\mathtt{T}$ vary from $10$ to $50$. The result and the runtime of each algorithm are shown in Fig. (\ref{fig:gnutellaDelay1_size}), (\ref{fig:gnutellaDelay4_size}), (\ref{fig:gnutellaDelay1_time}) and (\ref{fig:gnutellaDelay4_time}). When the weight functions were all linear, we observed the same pattern as in the random graph, where the quality of solution of $\mathtt{IG}$, $\mathtt{AT}$ and $\mathtt{LR}$ almost overlapped. Actually, $\mathtt{LR}$ always returned the best solution but the gap between \texttt{LR} and $\mathtt{AT}$ and $\mathtt{IG}$ was insignificant. Meanwhile, the sizes of the $\mathtt{SA}$'s solutions were always within 1.3 factor from $\mathtt{LR}$. However, in terms of runtime, $\mathtt{AT}$ performed the worst among our proposed algorithms while $\mathtt{LR}$ again was the best. The next best algorithm in terms of runtime was $\mathtt{SA}$, which stayed within 1.4 factor from $\mathtt{LR}$. Starting from $\mathtt{T} = 35$, the runtime of $\mathtt{LR}, \mathtt{IG}$ and $\mathtt{SA}$ almost stayed the same while the runtime of $\mathtt{AT}$ and $\mathtt{CC}$ kept increasing. 

However, it was a different matter in the experiments with heterogeneous weight functions, shown in Fig. (\ref{fig:gnutellaDelay4_size}) and Fig. (\ref{fig:gnutellaDelay4_time}). In terms of the quality of solution, the ranking from best to worst was $\mathtt{AT}$, $\mathtt{IG}$, $\mathtt{SA}$ and $\mathtt{CC}$. These algorithms were now more virtually distinguishable in solution quality. The sizes of solutions of $\mathtt{IG}$ could be up to 1.25 factor from $\mathtt{AT}$'s while this number of $\mathtt{SA}$ was 4. In terms of runtime, $\mathtt{CC}$ was no longer the worst algorithm. Starting from $\mathtt{T} = 45$, $\mathtt{AT}$ ran slower than $\mathtt{SA}$ and $\mathtt{CC}$. $\mathtt{IG}$ was by far the fastest among the algorithms.

% \begin{figure}[t] 
% % \vspace*{-20px}
% \subfloat[Solution size]{\includegraphics[width=0.25\textwidth]{image/roadNet-CA_undirected_delay1_changeT_result.png}
% \label{fig:roadDelay1_size}}
% ~
% \hspace{-20px}
%  \subfloat[Running time]{
%   	\includegraphics[width=0.25\textwidth]{image/roadNet-CA_undirected_delay1_changeT_time.png}
%   	\label{fig:roadDelay1_time}}
    
%     \caption{RoadnetCA with linear weight functions}
%  	\label{fig:roadDelay1}
% \end{figure}

% \begin{figure}[t] 
% % \vspace*{-20px}
% \subfloat[Solution size]{\includegraphics[width=0.25\textwidth]{image/roadNet-CA_undirected_delay4_changeT_result.png}
% \label{fig:roadDelay4_size}}
% ~
% \hspace{-20px}
%  \subfloat[Running time]{
%   	\includegraphics[width=0.25\textwidth]{image/roadNet-CA_undirected_delay4_changeT_time.png}
%   	\label{fig:roadDelay4_time}}
    
%     \caption{RoadnetCA with heterogeneous weight functions}
%  	\label{fig:roadDelay4}
% \end{figure}

% \begin{figure}[t]
% % \vspace*{-20px}
% \subfloat[Solution size]{\includegraphics[width=0.25\textwidth]{image/roadNet-CA_undirected_delay5_changeT_result.png}
% \label{fig:roadDelay5_size}}
% ~
% \hspace{-20px}
%  \subfloat[Running time]{
%   	\includegraphics[width=0.25\textwidth]{image/roadNet-CA_undirected_delay5_changeT_time.png}
%   	\label{fig:roadDelay5_time}}
    
%     \caption{RoadnetCA with cutting weight functions}
%  	\label{fig:roadDelay5}
% \end{figure}

Next, we experimented our algorithms on large scale networks. The Roadnet network contains 2 millions of nodes but only 2.8 millions of edges, which made it the sparest network among the datasets we used for experiment. First, we varied the value of $\mathtt{T}$ from 100 to 120 and plotted the results as in Fig. (\ref{fig:roadDelay1_size}) and Fig. (\ref{fig:roadDelay1_time}). $\mathtt{SA}$ performed much worse than $\mathtt{LR}$, $\mathtt{AT}$ and $\mathtt{IG}$. Its sizes of solutions were always at least 3 times greater than the others' and roughly near $\mathtt{CC}$ when $\mathtt{T}$ increased. Although $\mathtt{CC}$ always returned the worst solution, it was by far the fastest. The second best in terms of runtime was $\mathtt{IG}$ but it was always at least 15 times slower than $\mathtt{CC}$. This number in $\mathtt{LR}$ and $\mathtt{AT}$ were 30 and 60 respectively. With $\mathtt{T}=100$, $\mathtt{SA}$ was slightly faster than $\mathtt{LR}$ and $\mathtt{AT}$, therefore, we reduced the experimental range of $\mathtt{T}$ to $[70,100]$ on next experiment to observe the behaviors of $\mathtt{SA}$. Interestingly, $\mathtt{SA}$ performed much more faster than $\mathtt{AT}$ and $\mathtt{IG}$ in this range. 

In the experiment with heterogeneous weight functions, as can be seen in Fig. (\ref{fig:roadDelay4_size}) and Fig. (\ref{fig:roadDelay4_time}), $\mathtt{SA}$ ran up to 8 and 16 times faster than $\mathtt{IG}$ and $\mathtt{AT}$ respectively. However, $\mathtt{SA}$'s solutions were still worse than those two algorithms. From our observation, we found it hard to predict the behaviors of $\mathtt{SA}$ especially when the set $\mathcal{F}$ becomes larger. $\mathtt{SA}$ can perform well when this set is small and is very stable in a certain range of this set's size. But when it exceeds this range, $\mathtt{SA}$'s runtime increases at a higher rate than any other algorithms we have considered. 

Finally, we evaluated the algorithms on the cutting scenario and reported the results in Fig. (\ref{fig:roadDelay5_size}) and Fig. (\ref{fig:roadDelay5_time}). Up to $\mathtt{T} = 91$, \texttt{IG} and \texttt{AT} were the best in term of quality of solution but then were bypassed by \texttt{TAG}. In term of runtime, in most cases, our algorithms were 30 times slower than the fastest one, \texttt{SAP}.

\begin{figure}[t]
\subfloat[Linear]{\includegraphics[width=0.35\textwidth]{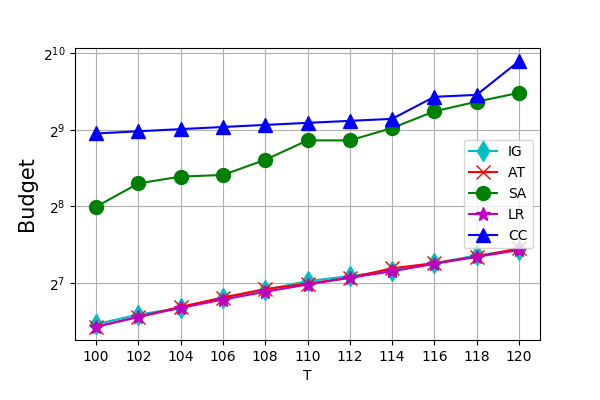}
\label{fig:roadDelay1_size}}
~
\hspace{-20px}
\subfloat[Heterogeneous]{\includegraphics[width=0.35\textwidth]{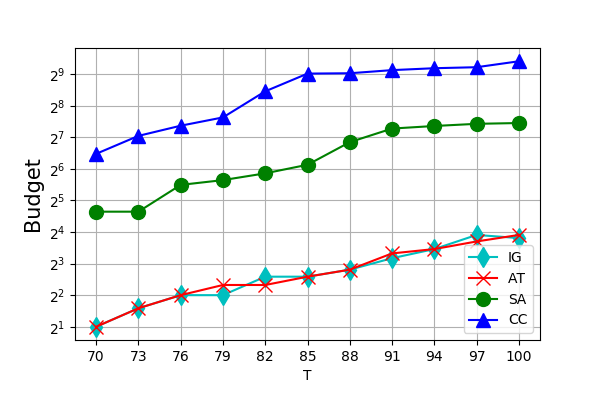}
\label{fig:roadDelay4_size}}
~
\hspace{-20px}
\subfloat[Cutting]{\includegraphics[width=0.35\textwidth]{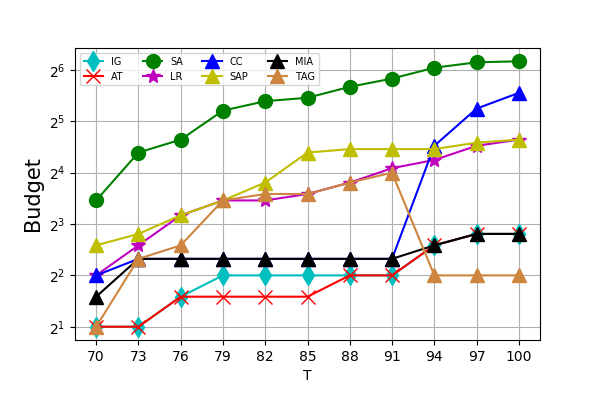}
\label{fig:roadDelay5_size}}
    
    \caption{Solution quality of algorithms on RoadnetCA}
 	\label{fig:road_size}
\end{figure}

\begin{figure}[t]
\subfloat[Linear]{\includegraphics[width=0.35\textwidth]{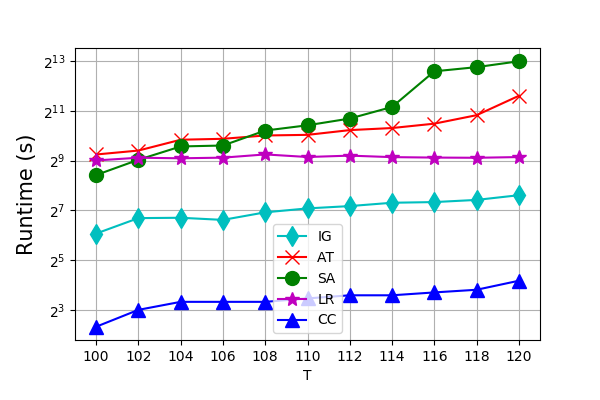}
\label{fig:roadDelay1_time}}
~
\hspace{-20px}
\subfloat[Heterogeneous]{\includegraphics[width=0.35\textwidth]{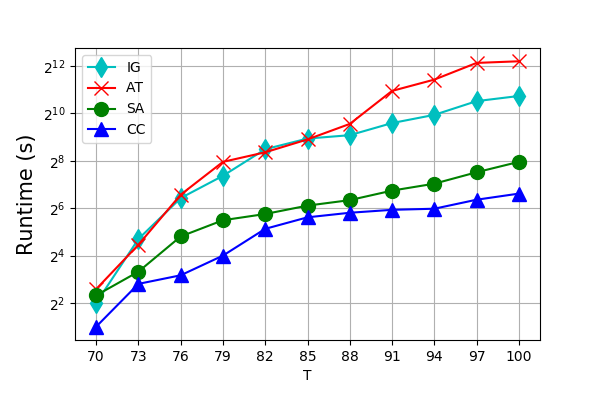}
\label{fig:roadDelay4_time}}
~
\hspace{-20px}
\subfloat[Cutting]{\includegraphics[width=0.35\textwidth]{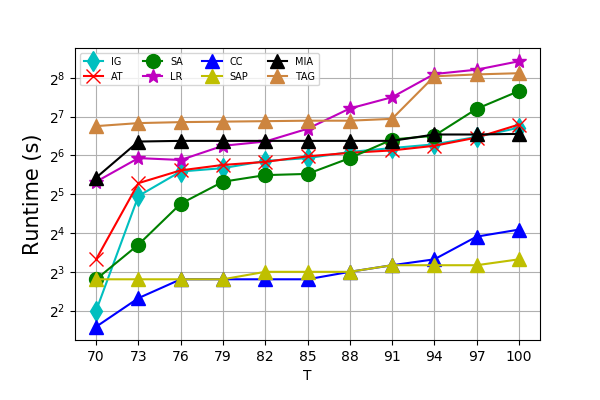}
\label{fig:roadDelay5_time}}
    
    \caption{Runtime of algorithms on RoadnetCA}
 	\label{fig:road_time}
\end{figure}

% We also compared our algorithms with three methods proposed by Kuhnle [?] for the \texttt{LB-MULTICUT} problem

% \begin{figure}[t] 
% % \vspace*{-20px}
% \subfloat[Solution size]{\includegraphics[width=0.25\textwidth]{image/as-skitter_undirected_delay4_changeT_result.png}
% \label{fig:skitterDelay4_size}}
% ~
% \hspace{-20px}
%  \subfloat[Running time]{
%   	\includegraphics[width=0.25\textwidth]{image/as-skitter_undirected_delay4_changeT_time.png}
%   	\label{fig:skitterDelay4_time}}
    
%     \caption{Skitter with heterogeneous weight functions}
%  	\label{fig:skitterDelay4}
% \end{figure}

The last network we did experiments on was Skitter, which is a dense graph where the average degree of a node is 6.5. In this experiment, we varied $\mathtt{T}$ in the range from $5$ to $12$. Fig. (\ref{fig:skitterDelay4_size}) and Fig. (\ref{fig:skitterDelay4_time}) show the performance of our algorithms. $\mathtt{IG}$, $\mathtt{AT}$ and $\mathtt{SA}$ could finish within the limited runtime while $\mathtt{CC}$ was unable to run even at $\mathtt{T} = 7$, which is why we did not show $\mathtt{CC}$'s results from $\mathtt{T} = 7$ in those figures. Also, this experiment clearly shows the trade-off between $\mathtt{IG}$ and $\mathtt{AT}$. The solution of $\mathtt{IG}$ was up to 1.25 times of $\mathtt{AT}$ while running faster with almost the same factor.

\subsubsection{Summary of results} The experimental results can be summarized as follows.
\begin{itemize}
\item In case of linear weight functions, $\mathtt{LR}$ always returned the best solution. The solution quality of $\mathtt{IG}$ and $\mathtt{AT}$ were worse than $\mathtt{LR}$ but usually by only a small factor. In addition, $\mathtt{IG}$ usually ran the fastest while the runtime of $\mathtt{AT}$ was more impacted by the varying of $\mathtt{T}$ than the other two.
\item In general cases where $\mathtt{LR}$ is no more applicable, $\mathtt{AT}$ was always the algorithm that returned the best quality of solution. $\mathtt{IG}$ was competitive to $\mathtt{AT}$ only if the weight functions tent to be more concave. In trade-off, $\mathtt{IG}$ performed much more faster than $\mathtt{AT}$ in most experiments.
\item In most experiments, $\mathtt{SA}$ was the worst among our algorithms in term of both the quality of solution and runtime. However, in several cases when the set of feasible paths was small or the input network was sparse, $\mathtt{SA}$ outperformed our other algorithms in runtime and the intuitive heuristics in solution quality. 
\end{itemize}

\section{Conclusion} \label{sec:conclusion}
In this work, we have introduced a new \texttt{QoSD} problem together with four solutions \texttt{IG, AT, SA} and \texttt{LR}, each of which scales to networks with millions of edges and nodes in under several hours and has a proven performance guarantee. Future work would include lowering the number of samples required by \texttt{SA}, making it more scalable. In addition, bounding the size of a set of candidate paths on \texttt{IG} and \texttt{AT} is necessary to reduce the burden on memory and waste of works when the candidate set is undesirable, and considering the correlation in increasing the edges' weights. Following that, we will investigate more on QoS degradation assessment on interdependent networks where networks are intertwined and interdependent, making the task of devising efficient algorithms much more challenging.

\begin{acks}
The authors would like to thank the anonymous reviewers for their valuable comments and helpful suggestions. We would also like to show our gratitude to Dr. Figueiredo (UFRJ) for shepherding our paper. This work is supported in part by NSF EFRI-1441231, NSF CNS-1814614, and DTRA HDTRA1-14-1-0055.
\end{acks}
% We target to improve our algorithms to finish on networks with billions of edges and nodes under a few hours.

\bibliographystyle{ACM-Reference-Format}
\bibliography{main_pomacs}

\clearpage
\appendix
\onecolumn
\appendix
\section*{APPENDIX}
\section*{Proof of Lemma \ref{lemma:concave}}
Let $\mathbf{s} = \{s_1,..s_m\}$, since $\mathbf{s}$ is a unit vector, there is only one value among $s_1,..s_m$ is $1$ and the others are all $0$. By extending the Equ. \ref{equ:delta_g}, we have
\begin{align*}
\Delta_\mathbf{s} g(\mathcal{P}, \mathbf{x}) = \sum_{p \in \mathcal{P}} \beta_p \Big( \min(\mathtt{T}, \sum_{e \in p} f_e(x_e+s_e)) - \min(\mathtt{T}, \sum_{e \in p} f_e(x_e)) \Big)
\end{align*}
For each path $p$, we will prove that:
\begin{align*}
& \min(\mathtt{T}, \sum_{e \in p} f_e(x_e+s_e)) - \min(\mathtt{T}, \sum_{e \in p} f_e(x_e)) \\
& \quad \geq \gamma \cdot (\min(\mathtt{T}, \sum_{e \in p} f_e(y_e+s_e)) - \min(\mathtt{T}, \sum_{e \in p} f_e(y_e))) 
\end{align*}
We consider three cases:
\begin{itemize}
\item $\mathtt{T} \geq  \sum_{e \in p} f_e(x_e+s_e) \geq  \sum_{e \in p} f_e(x_e)$. Then we have
\begin{align*}
&\min(\mathtt{T}, \sum_{e \in p} f_e(x_e+s_e)) - \min(\mathtt{T}, \sum_{e \in p} f_e(x_e)) \\
& \quad \quad = \sum_{e \in p} f_e(x_e+s_e) - \sum_{e \in p} f_e(x_e) \\ 
&\quad \quad \geq \gamma \cdot \Big( \sum_{e \in p} f_e(y_e+s_e) - \sum_{e \in p} f_e(y_e) \Big) \\
& \quad \quad \geq \gamma \cdot \Big(\min(\mathtt{T}, \sum_{e \in p} f_e(y_e+s_e)) - \min(\mathtt{T}, \sum_{e \in p} f_e(y_e)) \Big) 
\end{align*}
\item $\sum_{e \in p} f_e(x_e+s_e) \geq \mathtt{T} \geq  \sum_{e \in p} f_e(x_e)$. In this case,
\begin{align*}
\sum_{e \in p} f_e(y_e + s_e) \geq \sum_{e \in p} f_e(x_e + s_e) \geq \mathtt{T}
\end{align*}
also
\begin{align*}
min(\mathtt{T}, \sum_{e\in p}f_e(y_e)) \in [\sum_{e \in p} f_e(x_e), \mathtt{T}]
\end{align*}
Therefore:
\begin{align*}
&\min(\mathtt{T}, \sum_{e \in p} f_e(y_e+s_e)) - \min(\mathtt{T}, \sum_{e \in p} f_e(y_e)) \\
& \quad \quad = \mathtt{T} - \min(\mathtt{T}, \sum_{e \in p} f_e(y_e)) \\
& \quad \quad \leq \mathtt{T} - \sum_{e \in p} f_e(x_e) \\
& \quad \quad = \min(\mathtt{T}, \sum_{e \in p} f_e(x_e+s_e)) - \min(\mathtt{T}, \sum_{e \in p} f_e(x_e))
\end{align*}
\item $\sum_{e \in p} f_e(x_e+s_e) \geq  \sum_{e \in p} f_e(x_e) \geq \mathtt{T}$. This case is trivial because both $\min(\mathtt{T}, \sum_{e \in p} f_e(y_e+s_e)) - \min(\mathtt{T}, \sum_{e \in p} f_e(y_e))$ and $\min(\mathtt{T}, \sum_{e \in p} f_e(x_e+s_e)) - \min(\mathtt{T}, \sum_{e \in p} f_e(x_e))$ are $0$.
\end{itemize}
Hence, $\Delta_\mathbf{s} g(\mathcal{P}, \mathbf{x}) \geq \gamma \Delta_\mathbf{s} g(\mathcal{P}, \mathbf{y})$, which completes the proof.

\section*{Proof of Theorem \ref{theorem:greedy_approx}}
Denote $\mathbf{x}^*$ is optimal solution to the \texttt{QoSD} instance. Define $\mathbf{x}_i$ as our obtained solution before the $i^\textnormal{th}$ iteration in Alg. \ref{alg:greedy_blocking}. Denote $\mathbf{x}^o_i$ as an optimal solution that is in additional to $\mathbf{x}_i$ to block all paths in $\mathcal{P}$. We have:

\begin{align} \label{equ:x0x}
||\mathbf{x}^*|| \geq ||\mathbf{x}^* / \mathbf{x}_i|| \geq ||\mathbf{x}^o_i||
\end{align}

Assume $\mathbf{x}^o_i = \sum_{i=1}^l \mathbf{u}_i$ where $\mathbf{u}_i$ is a unit vector. We have:
\begin{align}
&\mathtt{D}(\mathcal{P}, \mathbf{x}_i + \mathbf{x}^o_i) - \mathtt{D}(\mathcal{P}, \mathbf{x}_i) = \sum_{j=1}^l \Delta_{\mathbf{u}_j} \mathtt{D}(\mathcal{P}, \mathbf{x}_i + \sum_{z=1}^{j-1} \mathbf{u}_z) \\
& \quad \quad \leq \frac{1}{\gamma} \sum_{j=1}^l \Delta_{\mathbf{u}_j} \mathtt{D}(\mathcal{P}, \mathbf{x}_i) \quad \quad \quad (\textnormal{Lemma \ref{lemma:concave}})\\
& \quad \quad \leq \frac{||\mathbf{x}^o_i||}{\gamma} \max_\mathbf{s} \Delta_\mathbf{s} \mathtt{D}(\mathcal{P}, \mathbf{x}_i) \\
& \quad \quad \leq \frac{\mathtt{OPT}}{\gamma} (\mathtt{D}(\mathcal{P}, \mathbf{x}_{i+1}) - \mathtt{D}(\mathcal{P}, \mathbf{x}_i)) \label{equ:hmm}\\
& \quad \quad = \frac{\mathtt{OPT}}{\gamma} (|\mathcal{P}|\mathtt{T} - \mathtt{D}(\mathcal{P}, \mathbf{x}_i) - (|\mathcal{P}| \mathtt{T} - \mathtt{D}(\mathcal{P}, \mathbf{x}_{i+1})))
\end{align}

Equ. \ref{equ:hmm} follows by greedy selection. Since $\mathtt{D}(\mathcal{P}, \mathbf{x}_1 + \mathbf{x}^o_i) = |\mathcal{P}| \mathtt{T}$,
\begin{align*}
|\mathcal{P}|\mathtt{T} - \mathtt{D}(\mathcal{P}, \mathbf{x}_{i+1}) \leq (1-\frac{\gamma}{\mathtt{OPT}}) (|\mathcal{P}|\mathtt{T} - \mathtt{D}(\mathcal{P}, \mathbf{x}_i))
\end{align*}

Note that the Alg. \ref{alg:greedy_blocking} will terminate after $||\mathbf{x}||$ iterations. Therefore: 

% Denote $k = ||\mathbf{x}||$ is the size of the solution obtained by algorithm \ref{alg:greedy_blocking}, which also means Alg. \ref{alg:greedy_blocking} will terminate after $k$ iterations. We have

\begin{align*}
&|\mathcal{P}| \mathtt{T} - \mathtt{D}(\mathcal{P}, \mathbf{x}_{||\mathbf{x}||}) \leq (1-\frac{\gamma}{\mathtt{OPT}}) (|\mathcal{P}| \mathtt{T} - \mathtt{D}(\mathcal{P}, \mathbf{x}_{||\mathbf{x}||-1})) \leq ... \\
& \quad \quad \leq (1-\frac{\gamma}{\mathtt{OPT}})^{||\mathbf{x}||} (|\mathcal{P}|\mathtt{T} - \mathtt{D}(\mathcal{P}, \{0\}^{|E|})) \\
& \quad \quad \leq (1-\frac{\gamma}{\mathtt{OPT}})^{||\mathbf{x}||} |\mathcal{P}| \mathtt{T}
\end{align*}

Since there should be at least a path $p \in \mathcal{P}$ whose overall delay is at most $\mathtt{T} - 1$ in final round, we have $|\mathcal{P}|\mathtt{T} - \mathtt{D}(\mathcal{P}, \mathbf{x}_l) \geq 1$. Therefore:

\begin{align*}
||\mathbf{x}|| \leq \frac{\ln |\mathcal{P}| \mathtt{T}}{\ln \frac{1}{1-\frac{\gamma}{\mathtt{OPT}}}} = \frac{\ln |\mathcal{P}| \mathtt{T}}{\ln (1 + \frac{\gamma/\mathtt{OPT}}{1 - \gamma/\mathtt{OPT}})}
\end{align*}

We have $\ln(1+x) \geq x - \frac{x^2}{2}$ for $x \in (0,1)$. So

\begin{align*}
||\mathbf{x}|| \leq \frac{\ln |\mathcal{P}| \mathtt{T}}{\frac{\gamma}{\mathtt{OPT}} (1 - \frac{\gamma}{2\mathtt{OPT}})} \leq \mathtt{OPT} \cdot O(\frac{\ln |\mathcal{P}| \mathtt{T}}{\gamma})
\end{align*}

And since $|\mathcal{P}| \leq n^\mathtt{h}$, \texttt{IG} obtains $O(\frac{1}{\gamma} (\mathtt{h}\ln n + \ln \mathtt{T}))$ approximation guarantee, which completes the proof.

\section*{Proof of Theorem \ref{theorem:trunk_approx}}
Denote $\mathbf{x}^* = \{x_1^*, ... x_m^*\}$ as optimal solution to the \texttt{QoSD} problem. Define $\mathbf{x}_i = \{x_1,...x_m\}$ is our obtained solution before the $i^\textnormal{th}$ iteration in Alg. \ref{alg:trunk_adding}. Denote $\mathbf{x}^o_i = \{x_1^o, ... x_m^o\}$ as an optimal solution in additional to $\mathbf{x}_i$ to block all paths in $\mathcal{P}$. We have:

\begin{align*}
||\mathbf{x}^*|| \geq ||\mathbf{x}^* / \mathbf{x}_i|| \geq ||\mathbf{x}^o_i||
\end{align*}

Denote $\mathbf{v}(e) = \{x_1,..x_{e-1},x_e + x_e^o,...x_m + x_m^o\}$. Trivially, $\mathbf{v}(1) = \mathbf{x}_i + \mathbf{x}^o$ and $\mathbf{v}(m+1) = \mathbf{x}_i$. Assume $\mathbf{u}(e_i,j_i)$ is the vector we would add into solution $\mathbf{x}_i$ in iteration $i^\textnormal{th}$. We have following lemma.

\begin{lemma}
For all $e \in E$, we have:
\begin{align*}
\frac{\Delta_{\mathbf{u}(e_i,j_i)} \mathtt{D}(\mathcal{P},\mathbf{x}_i)}{j_i} \geq \frac{\mathtt{D}(\mathcal{P},\mathbf{v}(e)) - \mathtt{D}(\mathcal{P},\mathbf{v}(e+1))}{x_e^o}
\end{align*} 
\end{lemma}

\begin{proof}
Denote $\mathbf{w}(e) = \{x_1,...x_{e-1}, x_e + x_e^o, x_{e+1}, ... x_m\}$. Consider a single path $p \in \mathcal{P}$, denote
\begin{align*}
&h(p,s) =  \sum_{e \in p \& e < s} f_e(x_e) + \sum_{e \in p \& e \geq s} f_e(x_e + x_e^o) \\
&g(p,s) = \sum_{e \in p \& e \neq s} f_e(x_e) + f_e(x_s + x_s^o)
\end{align*}
then we have:
\begin{align*}
&\mathtt{r}(p,\mathbf{v}(s)) - \mathtt{r}(p,\mathbf{v}(s+1)) = \min(\mathtt{T}, h(p,s)) - \min(\mathtt{T}, h(p,s+1)) \\
&\mathtt{r}(p,\mathbf{w}(s)) - \mathtt{r}(p,\mathbf{x}_i) = \min(\mathtt{T}, g(p,s)) - \min(\mathtt{T}, \sum_{e \in p} f_e(x_e))
\end{align*}

Trivially, we have that:
\begin{align*}
&h(p,s) - h(p, s+1) = g(p,s) - \sum_{e\in p} f_e(x_e) = f_{s}(x_{s} + x_{s}^o) - f_{s}(x_{s})
\end{align*}
and due to monotonicity of $\mathtt{r}(p, \mathbf{x})$
\begin{align*}
&h(p,s) \geq g(p,s) \\
&h(p,s+1) \geq \sum_{e \in p} f_e(x_e)
\end{align*}
Therefore, using the similar proof as lemma \ref{lemma:concave}, we have:
\begin{align*}
&\mathtt{D}(\mathcal{P},\mathbf{v}(e)) - \mathtt{D}(\mathcal{P},\mathbf{v}(e+1)) \leq \mathtt{D}(\mathcal{P},\mathbf{w}(e)) - \mathtt{D}(\mathcal{P}, \mathbf{x}_i) \\
& \quad \quad \quad = \Delta_{\mathbf{u}(e, x_{e}^o)} \mathtt{D}(\mathcal{P}, \mathbf{x}_i)
\end{align*}
Due to \texttt{AT} selection, we have that:
\begin{align*}
\frac{\Delta_{\mathbf{u}(e, x_{e}^o)} \mathtt{D}(\mathcal{P}, \mathbf{x}_i)}{x_{e}^o} \leq \frac{\Delta_{\mathbf{u}(e_i,j_i)} \mathtt{D}(\mathcal{P},\mathbf{x}_i)}{j_i}
\end{align*}
in which the lemma follows.
\end{proof}

Now, we will find the approximation guarantee of \texttt{AT} solution. We have:
\begin{align*}
&\mathtt{D}(\mathcal{P}, \mathbf{x}_i + \mathbf{x}^o_i) - \mathtt{D}(\mathcal{P}, \mathbf{x}_i) = \sum_{e} (\mathtt{D}(\mathcal{P}, \mathbf{v}(e)) - \mathtt{D}(\mathcal{P}, \mathbf{v}(e+1))) \\
& \quad \quad \quad \leq \sum_e \frac{x_e^o}{j_i} \Delta_{\mathbf{u}(e_i,j_i)} \mathtt{D}(\mathcal{P},\mathbf{x}_i) \\
& \quad \quad \quad \leq \frac{\mathtt{OPT}}{j_i} (\mathtt{D}(\mathcal{P}, \mathbf{x}_{i+1}) - \mathtt{D}(\mathcal{P}, \mathbf{x}_i))
\end{align*}
Since $\mathtt{D}(\mathcal{P}, \mathbf{x}_i + \mathbf{x}_i^o) = |\mathcal{P}| \mathtt{T}$, we have
\begin{align*}
|\mathcal{P}|\mathtt{T} - \mathtt{D}(\mathcal{P}, \mathbf{x}_{i+1}) \leq (1 - \frac{j_i}{\mathtt{OPT}}) (|\mathcal{P}| \mathtt{T} - \mathtt{D}(\mathcal{P}, \mathbf{x}_i))
\end{align*}
Assume \texttt{AT} stops after $l$ iterations, we have
\begin{align}
&|\mathcal{P}|\mathtt{T} - \mathtt{D}(\mathcal{P}, \mathbf{x}_{l}) \leq \prod_{i=1}^l (1-\frac{j_i}{\mathtt{OPT}}) (|\mathcal{P}|\mathtt{T} - \mathtt{D}(\mathcal{P}, \{0\})) \\
&\quad \quad \leq \Big( 1 - \frac{\sum_{i=1}^l j_i}{l \cdot \mathtt{OPT} } \Big)^l (|\mathcal{P}|\mathtt{T} - \mathtt{D}(\mathcal{P}, \{0\})) \label{equ:cauchy_ta} \\
& \quad \quad \leq e^{-\frac{||\mathbf{x}||}{\mathtt{OPT}}} (|\mathcal{P}|\mathtt{T} - \mathtt{D}(\mathcal{P}, \{0\})) \label{equ:e_ta}
\end{align}
Equ. \ref{equ:cauchy_ta} comes from the following Cauchy theorem
\begin{theorem} \label{theorem:cauchy}
(\textnormal{Cauchy Theorem} \cite{cauchyInequality}) Given $n$ non-negative numbers $x_1,...x_n$, we have
\begin{align*}
\prod_{i=1}^n x_i \leq (\frac{\sum_{i=1}^n x_i}{n})^n
\end{align*}
\end{theorem}
Equ. \ref{equ:e_ta} comes from observation that $(1-\frac{x}{n})^n \leq e^{-x}$.

Therefore, $||\mathbf{x}||_1 \leq \mathtt{OPT} \ln |\mathcal{P}|\mathtt{T}$. Since $|\mathcal{P}|$ is bounded by $n^\mathtt{h}$, \texttt{AT} obtains $O(\mathtt{h} \log n + \log \mathtt{T})$ approximation guarantee.

\section*{Proof of Lemma \ref{lem:greedy_sampling}}
Denote $\mathbf{v}_i$ as the budget vector $\mathbf{v}$ after greedily selecting first $i$ unit vectors, then by monotonicity $\hat{B}(\mathcal{P}, \mathbf{x} + \mathbf{v}^o) \leq \hat{D}(\mathcal{P}, \mathbf{v} + \mathbf{v}^o + \mathbf{v}_i)$. We have
\begin{align}
&\hat{B}(\mathcal{P}, \mathbf{x} + \mathbf{v}^o) \leq \hat{B}(\mathcal{P},\mathbf{x} + \mathbf{v}^o + \mathbf{v}_i) \\
& \quad = \hat{B}(\mathcal{P},\mathbf{x} + \mathbf{v}_i) + \sum_{j=1}^q \Delta_{\mathbf{u}_j} \hat{B}(\mathcal{P},\mathbf{x} + \mathbf{v}_i + \sum_{i=1}^{j-1} \mathbf{u}_i) \\ 
& \quad \leq \hat{B}(\mathcal{P},\mathbf{x} + \mathbf{v}_i) + \frac{1}{\gamma} \sum_{j=1}^q \Delta_{\mathbf{u}_j}\hat{B}(\mathcal{P},\mathbf{x} + \mathbf{v}_i) \\
& \quad \leq \hat{B}(\mathcal{P},\mathbf{x} + \mathbf{v}_i) + \frac{q}{\gamma} (\hat{B}(\mathcal{P},\mathbf{x} + \mathbf{v}_{i+1}) - \hat{B}(\mathcal{P},\mathbf{x} + \mathbf{v}_i)) \label{equ:greedy_select}
\end{align}
The inequality (\ref{equ:greedy_select}) is due to greedy selection. Therefore,
\begin{align*}
&\hat{B}(\mathcal{P},\mathbf{x} + \mathbf{v}_{i+1}) - \hat{B}(\mathcal{P},\mathbf{x} + \mathbf{v}_i) \geq \frac{\gamma}{q}(\hat{B}(\mathcal{P},\mathbf{x} + \mathbf{v}^o) - \hat{B}(\mathcal{P},\mathbf{x} + \mathbf{v}_i)) 
\end{align*}
Which also means
\begin{align*}
&\hat{B}(\mathcal{P},\mathbf{x} + \mathbf{v}^o) - \hat{B}(\mathcal{P},\mathbf{x} + \mathbf{v}_{i+1}) \\
& \quad \leq (1-\frac{\gamma}{q}) (\hat{B}(\mathcal{P},\mathbf{x} + \mathbf{v}^o) - \hat{B}(\mathcal{P},\mathbf{x} + \mathbf{v}_{i}))
\end{align*}
Therefore
\begin{align*}
\hat{B}(\mathcal{P},\mathbf{x} + \mathbf{v}^o) - \hat{B}(\mathcal{P},\mathbf{x} + \mathbf{v}) \leq (1-\frac{\gamma}{q})^q (\hat{B}(\mathcal{P},\mathbf{x} + \mathbf{v}^o) - \hat{B}(\mathcal{P},\mathbf{x}))
\end{align*}
So
\begin{align*}
&\Delta_{\mathbf{v}}\hat{B}(\mathcal{P},\mathbf{x}) \geq (1-(1-\frac{\gamma}{q})^q) \Delta_{\mathbf{v}^o}\hat{B}(\mathcal{P},\mathbf{x}) \\
&\quad \quad \geq (1-e^{-\gamma}) \Delta_{\mathbf{v}^o}\hat{B}(\mathcal{P},\mathbf{x})
\end{align*}
which completes the proof.

\section*{Proof of Theorem \ref{theorem:sampling_approx}}
First, considering the greedy selection $\mathbf{v}$ in each sampling iteration, from lemma \ref{cor:greedy_sa}, we have
\begin{align*}
\Delta_\mathbf{v} B(\mathbf{x}) \geq (1-e^{-\gamma})(1-\epsilon) \Delta_{\mathbf{v}^*} B(\mathbf{x})
\end{align*}
with probability at least $1-\delta/||\mathtt{b}||$.

Denote $\mathbf{x}^o = \sum_i^s \mathbf{u}_i$ as an optimal solution, which is additional to $\mathbf{x}$, can block all paths in $\mathcal{F}$ ($\mathbf{u}_i$ is a unit vector). Let split $\mathbf{x}^o$ into $l = \ceil{\frac{||\mathbf{x}^o||}{q}}$ parts $L_1,...L_l$ where $L_i = \sum_{j=(i-1)q+1}^{iq} \mathbf{u}_j$, we have: 

\begin{align*}
&\Delta_{\mathbf{x}^o} B(\mathbf{x}) = \sum_{j=1}^{l} \Delta_{L_j} B(\mathbf{x} + L_1 + ... + L_{j-1}) \leq \frac{1}{\gamma} \sum_{j=1}^{l} \Delta_{L_j} B(\mathbf{x})
\end{align*}
Therefore, there should be at least a value $\Delta_{L_j} B(\mathbf{x}) \geq \frac{\gamma}{l} \Delta_{\mathbf{x}^o} B(\mathbf{x})$. And since $||L_j|| \leq q$, we have:

\begin{align*}
\Delta_{\mathbf{v}} B(\mathbf{x}) \geq \frac{\gamma}{l} (1-e^{-\gamma}) (1-\epsilon) \Delta_{\mathbf{x}^o} B(\mathbf{x}) 
\end{align*}
which also means
\begin{align*}
|\mathcal{F}|\mathtt{T} - B(\mathbf{x}+\mathbf{v}) \leq (1 - \frac{q\gamma}{\mathtt{OPT}} (1-e^{-\gamma})(1-\epsilon)) (|\mathcal{F}|\mathtt{T} - B(\mathbf{x}))
\end{align*}

Now, denote $\mathbf{x}_i$ as our solution after the $i^\textnormal{th}$ iteration of Alg. \ref{alg:sampling_solution}. We have

\begin{align*}
|\mathcal{F}|\mathtt{T} - B(\mathbf{x}_{i+1}) \leq (1 - \frac{q\gamma}{\mathtt{OPT}} (1-e^{-\gamma})(1-\epsilon)) (|\mathcal{F}|\mathtt{T} - B(\mathbf{x}_i))
\end{align*}

Assume the algorithm terminates after $g$ iterations, we have:
\begin{align*}
|\mathcal{F}|\mathtt{T} - B(\mathbf{x}_g) &\leq (1 - \frac{q\gamma}{\mathtt{OPT}} (1-e^{-\gamma})(1-\epsilon)) (|\mathcal{F}|\mathtt{T} - B(\mathbf{x}_{g-1})) \leq ... \\
& \leq \Big(1 - \frac{q\gamma}{\mathtt{OPT}} (1-e^{-\gamma})(1-\epsilon)\Big)^g (|\mathcal{F}|\mathtt{T} - B(\{0\}^m)) 
\end{align*}
Each inequality happens with probability at least $1- \frac{\delta}{||\mathtt{b}||}$. So the probability such that $|\mathcal{F}|\mathtt{T} - B(\mathbf{x}_g) \leq \Big(1 - \frac{q\gamma}{\mathtt{OPT}} (1-e^{-\gamma})(1-\epsilon)\Big)^g (|\mathcal{F}|\mathtt{T} - B(\{0\}^m))$ is at least $1 - \frac{\delta g}{||\mathtt{b}||} \geq 1 - \delta$. Moreover, since in the $g^\textnormal{th}$ iteration, there should exist a path $p \in \mathcal{F}$, whose length smaller than $\mathtt{T}$. So, the maximum length of $p$ is $\mathtt{T}-1$. Therefore

\begin{align*}
1 \leq (1 - \frac{q\gamma}{\mathtt{OPT}} (1-e^{-\gamma})(1-\epsilon))^g |\mathcal{F}|\mathtt{T}
\end{align*}

So $g \leq O(\frac{\ln \mathtt{T} + h \ln d}{q \gamma (1-e^{-\gamma})(1-\epsilon)}) \mathtt{OPT}$. Since in each iteration, a budget vector $\mathbf{v}$, $||\mathbf{v}|| \leq q$ is added into solution, out final solution guarantees $O(\frac{\ln \mathtt{T} + h \ln d}{\gamma (1-e^{-\gamma})(1-\epsilon)})$ approximation ratio with probability at least $1-\delta$.

\end{document}